\documentclass[letterpaper]{article} % DO NOT CHANGE THIS
\usepackage[preprint]{aaai2027}  
\usepackage[hyphens]{url}  % DO NOT CHANGE THIS
\usepackage{graphicx} % DO NOT CHANGE THIS
\usepackage{natbib}  % DO NOT CHANGE THIS AND DO NOT ADD ANY OPTIONS TO IT
\usepackage{caption} % DO NOT CHANGE THIS AND DO NOT ADD ANY OPTIONS TO IT
\usepackage{algorithm}
\usepackage{algorithmic}

\usepackage{newfloat}
\usepackage{listings}
\DeclareCaptionStyle{ruled}{labelfont=normalfont,labelsep=colon,strut=off} % DO NOT CHANGE THIS
\floatstyle{ruled}
\newfloat{listing}{tb}{lst}{}
\floatname{listing}{Listing}

\usepackage{booktabs}

\usepackage{amssymb} 
\usepackage{xcolor}
\usepackage{enumitem}  
\usepackage{tikz}  
\usepackage{tikz-cd}  
\usetikzlibrary{matrix,fit,decorations.pathmorphing,intersections,arrows,positioning,shapes.misc}  
\usepackage{mathtools} 
\usepackage{textcomp} 
\usepackage{stmaryrd}    

\newcommand{\MV}{{\small \textsf{MV}}}
\newcommand{\MVArg}{{\small \textsf{MVArg}}}
\newcommand{\residual}{\mathsf{residual}}

\newcommand{\eligible}[2]{\mathsf{eligible}(#1, #2)}
\newcommand{\hide}[1]{}
\newtheorem{definition}{Definition}    
\newtheorem{fact}{Fact}
\newtheorem{observation}{Observation} 
\newtheorem{theorem}{Theorem} 
\newtheorem{proposition}{Proposition} 
\newenvironment{proof}{%
  \par\noindent\textbf{Proof.}\quad% 行頭の文字（太字のProof.）
}{%
  \hfill$\square$\par% 末尾に証明終了マーク（白抜き四角）を置いて改行
}
\newtheorem{lemma}{Lemma} 
\newtheorem{example}{Example}

\title{Multi-Winner Voting with Argumentative Ballots}
\author {
    Ryuta Arisaka\textsuperscript{\rm 1}\corresponding 
    and
    Hirotaka Ono\textsuperscript{\rm 2}
}
\affiliations {
    \textsuperscript{\rm 1} Department of Informatics, Kyoto University, Kyoto, Japan\\
    \textsuperscript{\rm 2} Department of Mathematical Informatics, Nagoya University, Aichi, Japan\\
    ryutaarisaka@gmail.com, ono@i.nagoya-u.ac.jp}
\begin{document}

\maketitle

\begin{abstract} 
	We introduce {\it multi-winner voting with argumentative ballots} (MVArg) 
	and investigate theoretical properties.  
	As our {\it conceptual contribution}, we generalise 
	approval ballots to argumentative ballots, thereby  
	allowing voters 
	to express defeasible preferences over candidates.  
	We accordingly generalise voter cohesion and 
	justified representation axioms $\mathsf{JR}, \mathsf{PJR}$
	and $\mathsf{EJR}$. 
	As our {\it theoretical contribution}, we establish 
	several key results. {\it First}, MVArg   
	is strictly more expressive than multi-winner voting with approval ballots (MV). 
	{\it Second}, our notions of cohesion and justified representation 
	are conservative generalisations of their counterparts in MV. 
	{\it Third}, the MVArg counterpart of $\mathsf{JR}$
	can always be satisfied, whereas the counterparts of $\mathsf{PJR}$
	and $\mathsf{EJR}$ cannot always be. {\it Fourth}, although verifying 
	whether a winner set satisfies the 
	MVArg counterpart of $\mathsf{JR}$ is already $\mathsf{coNP}$-hard, such a winner set 
	can be constructed in polynomial time.   
	All definitions, propositions, auxiliary lemmas and theorems have been formalised and mechanically checked 
	in Lean 4. 
\end{abstract}


\section{1 Introduction}   \label{sec_introduction}  
From parliamentary elections and participatory budgeting to recommendation formation, collective decision making requires selecting 
multiple winners. 
A prominent model for such tasks is {\it approval-based 
multi-winner voting}, where each voter 
records a set of acceptable candidates and a voting rule aggregates the resulting 
approval profile into 
a fixed-size set of winners.   

A central concern in approval-based multi-winner voting is {\it proportional 
representation}. % is a key to 
Intuitively, a sufficiently large set of voters that agrees on 
sufficiently many candidates should see its preferred candidates represented 
in the winning set, in proportion to 
the group's degree of cohesion. This idea has led both to methods 
for selecting winner sets \cite{Lackner23} and to 
{\it justified representation axioms}, such as 
$\mathsf{JR}, \mathsf{PJR}$ and $\mathsf{EJR}$ \cite{Aziz17,Sanchez-Fernandez26}, 
which provide qualitative fairness criteria for evaluating such methods.  
These ideas have also influenced real participatory budgeting processes 
\cite{Peters20}. 

Notwithstanding this success, a recent social 
experiment suggests that voters perceive more expressive 
input formats as better reflecting their preferences than 
approval voting \cite{Yang24}. % than 
%approval voting. 
Full rankings are, of course, unrealistic both computationally and operationally 
in many settings \cite{Lang16}. 
Nevertheless, finding an alternative format that allows voters to express their 
preferences  
in more detail remains desirable. %\cite{Jain20}.  

With this motivation, we propose {\it argumentative ballots}  
for multi-winner voting. The underlying principle is 
that of 
{\it abstract argumentation frameworks} \cite{Dung95},  
originally introduced for reasoning about 
argumentative dialogues. % as a directed graph. % and for reasoning about 
%acceptability.  
An abstract argumentation framework models argumentative information 
as a directed graph of entities. It encodes {\it attack} 
and {\it defence} relationships among entities. 
As a very simple example, 
$c_3 \rightarrow c_2 \rightarrow c_1$ expresses that $c_3$ attacks $c_2$ which 
attacks $c_1$. Normally \cite{Dung95}, $c_3$ then defends $c_1$ and 
both $c_3$ and $c_1$ (but not $c_2$) are {\it acceptable}. 
In multi-winner voting with argumentative ballots, every voter receives a ballot 
on which the candidates are written. 
Instead of marking each candidate as approved or not approved 
as in an ordinary approval ballot, the voter 
can draw arrows between candidates to express  
attack claims. Approval by a set of voters 
is determined with respect to the combined directed graph 
obtained from their attack claims. 
%The protocol of multi-winner 
%voting with argumentative voting is 
%that candidate 

\[
\begin{tikzpicture}[
  baseline,
  box/.style={
    draw,
    rounded corners,
    inner sep=8pt
  }
]

\node[box] (B1) {
\begin{tikzcd}[row sep=small,column sep=small]
	c_3   & c_2 \arrow[r] & c_1  
\end{tikzcd}
};

\node[above left=2pt and 2pt of B1.north west, anchor=south west] 
	{\textbf{Alice's ballot}};

\node[box, right=20mm of B1] (B2) {
\begin{tikzcd}[row sep=small,column sep=small] 
	c_3 \arrow[r] & c_2 & c_1
\end{tikzcd}
};

\node[above left=2pt and 2pt of B2.north west, anchor=south west] 
	{\textbf{Bob's ballot}};

\end{tikzpicture}
\] 

With these two argumentative ballots, Alice approves both $c_3$ and $c_2$ 
but not $c_1$. Bob approves both $c_3$ and $c_1$ but not $c_2$.   
Collectively, Alice and Bob approve $c_3$ and $c_1$ but not $c_2$. 
The corresponding ordinary approval ballots 
would yield collective approval 
only of $c_3$ (from the individual 
approvals $\{c_3, c_2\}$ and $\{c_3,c_1\}$). 
So, argumentative ballots allow 
collective approvals to be more expressive: 
the same individual approval sets may give rise to different collective 
approvals depending on the underlying attack and defence relations. This additional 
expressiveness raises the central questions of this paper: how should voter cohesion 
and justified representation be defined when approvals are induced argumentatively, 
and to what extent are the computational advantages of approval-based multi-winner voting 
preserved in this setting? 

After technical preliminaries in Section 2, we make 
the following contributions. % in particular. 
\begin{enumerate}  
	\item We introduce {\it multi-winner voting with 
		argumentative ballots (\MVArg)} 
		(\textbf{Section 3.1}). 
	\item We show that 
		{\MVArg} is a conservative generalisation of   
		the standard approval-based multi-winner voting 
		({\MV}) \cite{Lackner23}. At the same time, {\MVArg} is strictly 
		more expressive: for fixed individual approval sets, 
		{\MV} determines a unique collective approval outcome, whereas 
		{\MVArg} may induce different 
		collective approvals (\textbf{Section 3.2}). 
	\item We define notions of voter cohesion for {\MVArg} 
		and establish several 
		of their theoretical properties (\textbf{Section 3.3}). 
	\item We formulate the {\MVArg} counterparts of 
		$\mathsf{JR}, \mathsf{PJR}$ and $\mathsf{EJR}$ 
		(\textbf{Section 4.1}).
	\item We show 
		that they are conservative 
		generalisations of $\mathsf{JR}, \mathsf{PJR}$ and $\mathsf{EJR}$ 
  (\textbf{Section 4.2}).   
		A winner set providing (the fairness defined in) the 
		$\MVArg$'s 
		counterpart of ${\mathsf{JR}}$ always exists.  
		For the $\MVArg$ counterparts of 
		$\mathsf{PJR}$ and $\mathsf{EJR}$,  
		we identify additional cohesion constraints that 
		restore existence (\textbf{Section 4.3}).   
		Verifying whether a given winner set provides 
		(the fairness defined in) the $\MVArg$ counterpart of 
		$\mathsf{JR}$ is already $\mathsf{coNP}$-hard, even though 
		such a winner set can be constructed 
		in polynomial time (\textbf{Section 4.4}).  
\end{enumerate}  
Technically, we introduce a direct inductive covering argument 
for the existence theorems, which to the best of our knowledge 
is a novel proof strategy in the study of justified representation.

Together, these results show what remains valid, what changes, and what becomes 
computationally harder when approval ballots are replaced by argumentative ballots.  

A Lean 4 formalisation of all definitions, propositions (Propositions \ref{prop_approval_preserving_transformation} 
through \ref{prop_lifting}), auxiliary lemmas (Lemmas \ref{lem_collective_approval_gadgets} through 
\ref{lem_relationship_between_cohesion_argumentative_cohesion}) and theorems (Theorems \ref{thm_preservation_m_m_representations} 
through \ref{thm_polynomial_time_construction}), together with the accompanying Java implementation of 
$\mathsf{GreedyGrounded}$ for Theorem \ref{thm_polynomial_time_construction}, 
is provided as supporting material with the arXiv submission. 

\subsection{Related work}
{\it Approval-based multi-winner voting.}
\cite{Aziz17} introduced $\mathsf{JR}$ and $\mathsf{EJR}$ for approval-based 
multi-winner voting, and \cite{Sanchez-Fernandez26} introduced $\mathsf{PJR}$, which lies 
between them. These axioms formalise group fairness: sufficiently large groups of voters
that agree on sufficiently many candidates should be represented in the winning set.
In this classical setting, approvals are primitive and collective approvals are monotonic.
$\MVArg$ differs in that approvals are induced by voters' argumentative ballots and may
become non-monotonic when ballots are combined. Justified representation therefore needs 
to be reformulated for $\MVArg$. 

{\it Collective choice and participatory institutions.} 
In institutional settings 
such as participatory budgeting, 
a limited portfolio should provide representation to cohesive groups of 
citizens \cite{Peters20,Peters21,Brill23}. 
Approval-based models generally treat reported 
approvals as primitive. In policy decisions, however, alternatives   
may be evaluated relationally: one project may substitute 
for another \cite{Jain20}, undermine its rationale, or answer an objection to it. 
$\MVArg$ allows for retaining such relations. % during aggregation.  
While our contribution remains axiomatic and computational, 
$\MVArg$ offers a foundation for studying representation when 
collective evaluations depend on explicitly reported reasons.

{\it Structured and conditional preferences.}
$\MVArg$ is also related to voting in combinatorial domains \cite{Lang16} and conditional
preference languages such as CP-nets \cite{Boutilier04}, where outcomes are structured 
and preferences may depend on relations among issues or variables. The source of these dependences is 
different in $\MVArg$: dependencies arise from attack and defence relations
inside argumentative ballots, and candidate approvals are induced by acceptability
semantics. Specifically, instead of representing preferences over bundles 
of candidates, $\MVArg$ represents justificatory 
relations among individually selectable candidates. 
Accepting one candidate may undermine or defend another. 
Our concern is therefore not preference optimisation 
over structured outcomes, but proportional representation 
when the acceptability of candidates is reason-dependent.

{\it Argumentation and voting.}
Abstract argumentation frameworks were introduced by \cite{Dung95} as directed graphs of 
arguments and attacks, together with semantics determining acceptability. Several works 
combine argumentation and voting. \cite{Leite11} incorporate positive and negative votes
into argumentation frameworks; \cite{Awad17} study judgment aggregation over argument
labels in a fixed framework; and \cite{Bernreiter24} 
combine approval voting and abstract
argumentation to select representative acceptable sets of arguments. These works largely assume 
a fixed underlying argumentation framework and primitive votes or approvals over its
arguments. 
More recently, \cite{Susami26} have explicitly assumed 
local abstract argumentation frameworks and one global 
abstract argumentation framework. Still, their work 
does not materialise collective approvals and 
its objective remains to filter acceptability semantics. 
In contrast, $\MVArg$ uses 
an argumentation framework as each voter's ballot: approvals are derived from those ballots, and can therefore become 
non-monotonic under ballot combination. Rather than  
aggregating evaluations of arguments within a fixed framework, 
we generalise approval-based multi-winner voting by allowing 
argumentation frameworks as ballots.

 \begin{figure*}[t] 	\includegraphics[width=\textwidth]{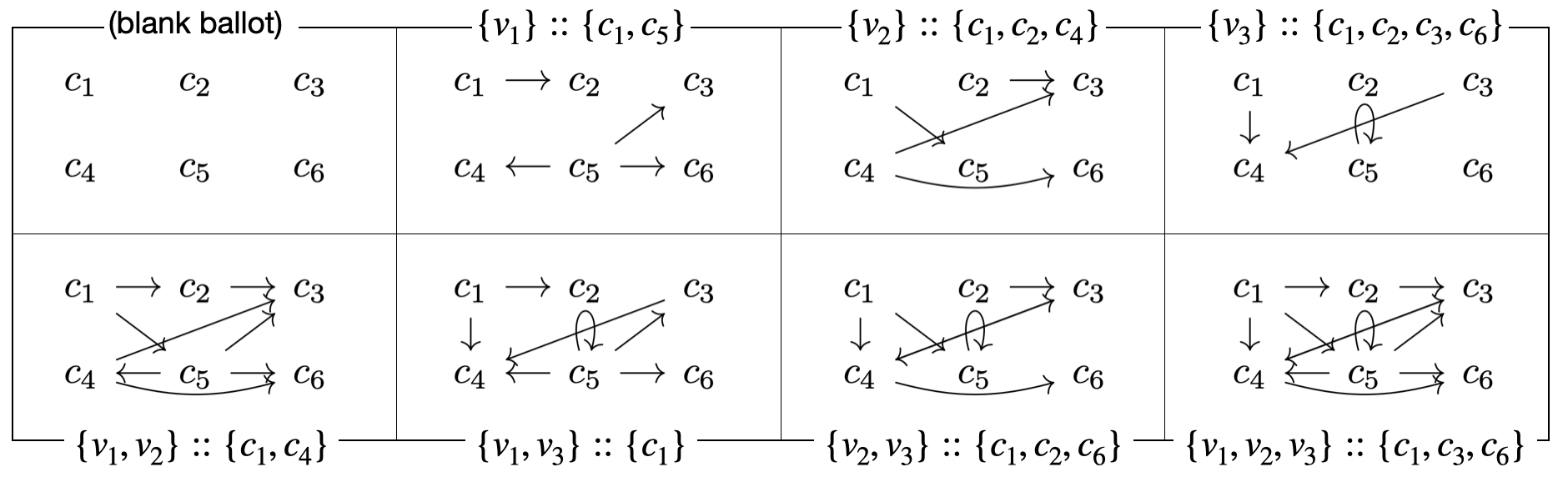} 
	 \caption{\textbf{The first row from left to right}: a blank ballot; 
	 $v_1$'s ballot approving $c_1$ and $c_5$; 
	 $v_2$'s ballot approving $c_1,c_2$ and $c_4$; and 
	 $v_3$'s ballot approving $c_1,c_2,c_3$ and $c_6$. 
	  \textbf{The second row from left to right}:    
	  $\{v_1,v_2\}$'s collective approval of 
	  $c_1$ and $c_4$; 
	  $\{v_1,v_3\}$'s collective approval of 
	  $c_1$; $\{v_2,v_3\}$'s collective approval 
	  of $c_1,c_2$ and $c_6$; 
	  and $\{v_1,v_2,v_3\}$'s collective approval of 
	  $c_1,c_3$ and $c_6$.} 
		 \label{fig_first}  
\end{figure*}

\section{2 Technical Preliminaries}  \label{sec_technical_preliminaries} 
\textbf{Multi-winner voting and justified representation.} 
Let $V \equiv \{v_1, \ldots, v_n\}$ be a set of voters, 
let $C \equiv \{c_1, \ldots, c_m\}$ be a set of candidates, 
let $A \equiv (A_{v_1}, \ldots, A_{v_n})$---where 
$A_{v_i}$ is a subset of $C$---be the voters' approval profile, 
and let $k$ be the number of winners not greater than $|C|$. 
Then, $(V, C, A, k)$ is a {\it multi-winner voting profile}. 
In this paper, $\mathfrak{M}^{\mathfrak{M}}$ denotes the set of all multi-winner voting profiles.

Let $A^{\forall}_{V'}$ denote $\bigcap_{v \in V'} A_v$. 
For any subset $V'$ of $V$ and any candidate $c \in C$, $V'$ {\it approves} $c$ iff (if and only if) 
$c \in A^{\forall}_{V'}$. So, $A^{\forall}_{V'}$ 
is the set of all candidates $V'$ approves. 
For any subset $V'$ of $V$ 
and any positive integer $l$, $V'$ is {\it $l$-cohesive} iff    
(1) 
$|V'| \geq l \cdot (|V|/k)$, and (2) $|A^{\forall}_{V'}| \geq l$.\footnote{
Here, $(|V|/k)$ is the {\it Hare Quota} for a subset of $V$ to rightly demand 1 winner 
allocation to it.} 

For any winner set $W$, which is a subset of $C$ with $|W| = k$, $W$ 
satisfies, or, following the terminology in \cite{Aziz17}, {\it provides}:
\begin{itemize} 
	\item $\mathsf{JR}$ iff, for any subset $V'$ of $V$, 
		if $V'$ is $1$-cohesive, then 
		there is some $v \in V'$ such that $|W \cap A^{\forall}_{\{v\}}| 
		\geq 1$.  
	\item $\mathsf{PJR}$ 
		iff, for any subset $V'$ of $V$, 
		if $V'$ is $l$-cohesive, then $|W \cap \bigcup_{v \in V'} A^{\forall}_{\{v\}}| \geq l$.  
	\item $\mathsf{EJR}$ 
		iff, for any subset $V'$ of $V$, 
		if $V'$ is $l$-cohesive, then there is some $v \in V'$ 
		such that $|W \cap A_{\{v\}}^{\forall}| \geq l$.  
\end{itemize} 
These are justified representation axioms.  
For $x \in \mathfrak{M}^{\mathfrak{M}}$, 
let $\nu$ be a member of $\{JR, PJR, EJR\}$, 
and let $\nu(x)$ be the set of all subsets $W$ of $C$ such that $W$ provides: 
$\mathsf{JR}$ if $\nu$ is $JR$; $\mathsf{PJR}$ if $\nu$ is $PJR$; and $\mathsf{EJR}$ if $\nu$ is $EJR$. Then: 
\begin{itemize} 
	\item $\emptyset \subset EJR(x) \subseteq 
PJR(x) \subseteq JR(x)$. 
	\item Some member of $\nu(x)$ is polynomial-time computable.  
	\item Verification of membership of $W$ in $EJR(x)$, 
as well as that in $PJR(x)$, is coNP-complete.  
	\item Verification of membership of $W$ in $JR(x)$ 
		is polynomial-time decidable. 
\end{itemize}   

\noindent \textbf{Grounded acceptance in abstract argumentation.} 
Let $(Arg, R)$ be an abstract argumentation framework \cite{Dung95} where 
$Arg$ is a finite set of entities (called arguments) and $R$ is a subset of $Arg \times Arg$. 
$a$ {\it attacks} $a'$ iff $(a, a') \in R$. 
Let $Arg'$ be a subset of $Arg$. $Arg'$ is {\it conflict-free} 
iff, for any $a_1 \in Arg'$ and $a_2 \in Arg'$, 
		$(a_1, a_2) \not\in R$. % holds. 
$Arg'$ {\it defends} an argument $a \in Arg$ iff,  
for any $a' \in  Arg$, if $(a', a) \in R$, then 
there is some $a'' \in Arg'$ such that $(a'', a') \in R$. 
$Arg'$ is {\it admissible} iff $Arg'$ is conflict-free and, 
for any $a \in Arg'$, $Arg'$ defends $a$.   
The 
{\it grounded extension} is then defined as a set of arguments 
$Arg'$ satisfying the following conditions. %strongly acceptable. 
\begin{itemize} 
	\item $Arg'$ is admissible and, 
		for any argument $a \in Arg$, if $Arg'$ defends $a$, then $a$ is in $Arg'$. 
	\item $Arg'$ is the least such set of arguments. 
\end{itemize} 
The following facts hold for every $(Arg, R)$ \cite{Dung95}. 
\begin{itemize}  
		\item Let $\mathcal{P}(...)$ denote the power set of $...$, the grounded extension is polynomial-time computable as the least 
		fixpoint of $F: \mathcal{P}(Arg) \rightarrow \mathcal{P}(Arg)$ 
		defined as: $F(Arg') = \{a \in Arg \mid Arg' \text{ defends } a\}$. 
	\item Every $(Arg, R)$ has a unique grounded extension. 
\end{itemize}  
Other types of extensions exist \cite{Baroni07}, but in 
general they encode more disputable acceptance, 
are not unique and can be 
computationally more demanding. 

\section{3 Multi-Winner Voting with Argumentative Ballots}  \label{sec_argumentative_voting} 
\subsection{3.1 Conceptualisation} 
As briefly described in Section 1, 
argumentative ballots allow voters to 
report their attack claims.  
\begin{example}[City projects]. \label{ex_city_planning}\rm
A city receives daily complaints about congestion on the main road between Districts A and B. It has proposed six projects to mitigate congestion:
\[
\begin{array}{ll}
c_1: & \text{build a new subway line},\\
c_2: & \text{widen the main road},\\
c_3: & \text{build a new tram line},\\
c_4: & \text{build high-frequency bus lanes},\\
c_5: & \text{introduce congestion pricing},\\
c_6: & \text{build a large parking garage}.
\end{array}
\]
Three projects will be selected. Voters $v_1,v_2,v_3$ cast argumentative ballots over 
these candidates. A blank ballot, shown in the first column of the first row of 
	Figure \ref{fig_first}, contains all candidates and no attacks; by default, they are all approved.  
	By drawing arrows, voters supply their attack claims. 
	$c_i \rightarrow c_j$ on a ballot 
	specifically means that acceptance of 
	$c_i$ undermines the case for $c_j$. 
	The ballots by $v_1$, $v_2$ and $v_3$ 
	are in the second, third and fourth columns of the first row. 
	\begin{description} 
		\item[$v_1$: ] Approves $c_1$ and $c_5$ with the 
			attack claims:  {\ }
			\begin{itemize} 
				\item ${c_1 \rightarrow c_2}$:   
					``A new subway line 
					undermines the case 
					for widening main road.'' 
		\item ${c_5 \rightarrow \{c_3, c_4, c_6\}}$: 
			``Congestion pricing undermines 
					the case for 
					a new tram 
					line, 
					high-frequency bus lanes and a large 
					parking garage.'' 
			\end{itemize} 
		\item[$v_2$: ] Approves $c_1, c_2$ and $c_4$ with 
			the attack claims: 
			{\ }
			\begin{itemize} 
				\item $c_1 \rightarrow c_5$, $c_2 \rightarrow c_3$ 
					and 
					$c_4 \rightarrow \{c_3, c_6\}$: 
					(All similarly) 
			\end{itemize}  
		\item[$v_3$: ] Approves $c_1, c_2, c_3, c_6$ with 
			the attack claims: {\ } 
			\begin{itemize} 
				\item $c_1 \rightarrow c_4$, $c_3 \rightarrow c_4$: (Similarly) 
				\item $c_5 \rightarrow c_5$: 
					``The case for congestion 
					pricing 
					is unconditionally 
					undermined.''
			\end{itemize}  
	       	\end{description}  
	Each voter individually approves the candidates in the grounded extension of the
	corresponding argumentation framework. For a set of voters, 
	the attack claims in their ballots are pooled together, %their 
%	defeasible preferences are
%	combined, 
	and the group 
	collectively approves the candidates in the grounded extension of the combined   
	framework, as shown in the second row.

		This produces non-monotonic collective approvals. For example, voter 
		$v_1$ alone does not approve $c_4$, because $v_1$ approves $c_5$ and
		$c_5$ attacks $c_4$. However, in the combined ballot of $\{v_1,v_2\}$, 
		$c_1$ is approved and $c_1$ attacks $c_5$. 
		The attack claim $c_5 
		\rightarrow c_4$ is therefore neutralised, 
		and $c_4$ becomes approved by $\{v_1,v_2\}$. 
		In other words, adding a voter's ballot can 
		turn a previously non-approved 
		candidate into an approved one. This is impossible in standard approval 
		ballots, where approvals are primitive and do not change when ballots 
		are combined.
		\hfill$\spadesuit$ 
\end{example} 
We now formalise this intuition  
by defining {\it argumentative voting profiles} and {\it approvals}. 
\begin{definition}[Argumentative voting profile] \rm    
	Let $V \equiv \{v_1, \ldots, v_n\}$ be a set of voters, 
	let $C \equiv \{c_1, \ldots, c_m\}$ be a set of candidates, 
	and let $B \equiv (B_{\{v_1\}}, \ldots, B_{\{v_n\}})$---where $B_{\{v_i\}}$ 
	is an abstract argumentation framework $(C, R_{v_i})$---be 
	voters' 
	argumentative ballots, and let $k$ be the number of winners not greater than 
	$|C|$. Then, $(V, C, B, k)$ is an {\it argumentative voting profile}. 
	$\mathfrak{M}^{\mathfrak{A}}$ denotes the set of all argumentative voting profiles.
	\hfill$\clubsuit$ 
\end{definition}   
{\it Example \ref{ex_city_planning}  
	gives rise to an argumentative voting profile 
	$(\{v_1, v_2,v_3\}, \{c_1, \ldots, c_6\}, B, 3)$ where $B$ is as given in the second, third and fourth columns of the first row 
	in Figure \ref{fig_first}.} 

Unless stated otherwise, we fix an arbitrary 
argumentative voting profile 
$(V, C, B, k)$. When we write $V$, it is the first component 
of the argumentative voting profile, and similarly for all other symbols. 
\begin{definition}[Approvals in argumentative voting profile]\rm 
	For any subset $V'$ of $V$ and any candidate $c \in C$,   
	$V'$ {\it approves} $c$ iff $c$  
	is in the grounded extension of $(C, \bigcup_{v \in V'}R_{v})$. 
	$B^{\forall}_{V'}$ denotes the set of all candidates approved by $V'$.  
	\hfill$\clubsuit$ 
\end{definition} 
{\it In Example 1, $B^{\forall}_{\{v_1\}}$ is $\{c_1, c_5\}$, $B^{\forall}_{\{v_2\}}$ 
is $\{c_1, c_2, c_4\}$, 
$B^{\forall}_{\{v_1, v_2\}}$ is $\{c_1, c_4\}$, $B^{\forall}_{\{v_1, v_2, v_3\}}$ is 
$\{c_1, c_3, c_6\}$ with the remaining collective approvals 
shown in Figure 1.} 

\subsection{3.2 Relationship between $\mathfrak{M}^{\mathfrak{M}}$ and $\mathfrak{M}^{\mathfrak{A}}$} 
A multi-winner voting profile is encoded  
into $\mathfrak{M}^{\mathfrak{A}}$ preserving approvals. 

\begin{proposition}[Approval-preserving transformation] \label{prop_approval_preserving_transformation} 
	There is some function $\tau: \mathfrak{M}^{\mathfrak{M}} \rightarrow \mathfrak{M}^{\mathfrak{A}}$ 
	such that, for any $(V, C, A, k) \in \mathfrak{M}^{\mathfrak{M}}$, $\tau(V, C, A, k) = 
	(V, C, B, k)$ (for some $B$)
	and that $A^{\forall}_{V'} = B^{\forall}_{V'}$ for any non-empty subset $V'$ of $V$. 
\end{proposition}   
$\tau$ can be instantiated in many ways, but  $(B_{\{v_1\}}, \ldots, B_{\{v_{|V|}\}})$ where 
$B_{\{v_i\}}$ is $(C, \{(c, c) \mid c \not\in A_{v_i}\})$ is one of the simplest. 
In this instantiation, an unconditional non-approval on the standard approval ballot 
is a self-loop around the candidate. 

{\it Continuing Example 1, let $(\{v_1, v_2, v_3\},\linebreak \{c_1, \ldots, c_6\}, A, 3)$ be such that, 
for any $v \in \{v_1, v_2, v_3\}$, $A_{\{v\}}^{\forall} = B_{\{v\}}^{\forall}$. 
The above concrete $\tau$ returns $(\{v_1, v_2, v_3\}, \{c_1, \ldots, c_6\}, B', 3)$ 
where $B'$ comprises the following argumentative ballots.  \\  
\[
\begin{tikzpicture}[
  baseline,
  box/.style={
    draw,
    rounded corners,
    inner sep=8pt
  }
]

\node[box] (B1) {
\begin{tikzcd}[row sep=small,column sep=small]
	c_1  & c_{2} \arrow[loop right]  \\
	c_3 \arrow[loop right] & c_4 \arrow[loop right] \\ 
	c_5  & c_6 \arrow[loop right] 
\end{tikzcd}
};

\node[above left=2pt and 2pt of B1.north west, anchor=south west] 
	{\(\boldsymbol{B'_{\{v_1\}}}\)};

\node[box, right=2mm of B1] (B2) {
\begin{tikzcd}[row sep=small,column sep=small]
  c_1  & c_{2} \\ 
	c_{3} \arrow[loop right] & c_{4} \\
	c_5 \arrow[loop right] & c_6 \arrow[loop right]
\end{tikzcd}
};

\node[above left=2pt and 2pt of B2.north west, anchor=south west] 
	{\(\boldsymbol{B'_{\{v_2\}}}\)};

\node[box, right=2mm of B2] (B12) {
\begin{tikzcd}[row sep=small,column sep=small]
  c_1  & c_{2}  \\
  c_{3} & c_{4} \arrow[loop right] \\
	c_5 \arrow[loop right] & c_6 
\end{tikzcd}
};

\node[above left=2pt and 2pt of B12.north west, anchor=south west] 
	{\(\boldsymbol{B'_{\{v_3\}}}\)};

\end{tikzpicture}
\]

{\ }\\ 
It holds that $B'^{\forall}_{\{v\}} = A^{\forall}_{\{v\}}$. 
} 

This approval-preserving transformation offers 
an argumentative interpretation of 
ballot combination in $\mathfrak{M}^{\mathfrak{M}}$: 
collective approval by a set of voters is the result of pooling 
the attack claims in their ballots and 
obtaining approved candidates in the resulting framework. 
Aggregation of attacks by union reproduces exactly the collective 
approvals in $\mathfrak{M}^{\mathfrak{M}}$.  

As for monotonicity of collective approvals in 
$\mathfrak{M}^{\mathfrak{M}}$: 
\begin{fact}[Approval monotonicity in $\mathfrak{M}^{\mathfrak{M}}$]\label{fact_approval_monotonicity} 
	Let $(V, C, A, k)$ be a multi-winner voting profile 
	in $\mathfrak{M}^{\mathfrak{M}}$. 
	If $V'_1$ and $V'_2$ are non-empty subsets of $V$ and if $V'_1 \subseteq V'_2$, 
		then $A^{\forall}_{V'_2} \subseteq A^{\forall}_{V'_1}$. 
\end{fact}       
$\tau$ preserves it in $\mathfrak{M}^{\mathfrak{A}}$, 
but, as evidenced in Figure \ref{fig_first}, 
it is not the universal property of 
$\mathfrak{M}^{\mathfrak{A}}$. In fact, for some argumentative voting profile, 
approval sets can strictly grow in size {\it e.g. $B^{\forall}_{\{v_1,v_3\}} 
\subset B^{\forall}_{\{v_1,v_2,v_3\}}$ in Figure 1}, for some, 
they can shrink in size, and for some other, 
the size of approval sets stays the same but their members change {\it e.g. 
$B^{\forall}_{\{v_1\}} \not= B^{\forall}_{\{v_1,v_2\}}$ but 
$|B^{\forall}_{\{v_1\}}| = |B^{\forall}_{\{v_1,v_2\}}| = 2$}.

\begin{proposition}[No bijection]  \label{prop_2}  
There is no function  
$\tau': \mathfrak{M}^{\mathfrak{A}} \rightarrow \mathfrak{M}^{\mathfrak{M}}$ 
	such that, for any $(V, C, B, k) \in \mathfrak{M}^{\mathfrak{A}}$, 
	$\tau'(V, C, B, k)$ is $(V, C, A, k)$ (for some $A$) 
	and that $B^{\forall}_{V'} = A^{\forall}_{V'}$ for any non-empty subset $V'$ of $V$. 
\end{proposition}  
Proposition \ref{prop_2}
clarifies the sense in which $\mathfrak{M}^{\mathfrak{A}}$ 
is more expressive than $\mathfrak{M}^{\mathfrak{M}}$. 
In $\mathfrak{M}^{\mathfrak{M}}$, once the individual approval 
sets are fixed, the collective approval of every voter set is 
fixed as their intersection. In $\mathfrak{M}^{\mathfrak{A}}$, 
by contrast, two voting profiles may induce the same 
individual approvals but different collective approvals, 
because attacks from different ballots can interact through 
defence after combination. 
Argumentative ballots therefore retain relational information 
that is lost in ordinary ballots. Propositions 
\ref{prop_approval_preserving_transformation} and \ref{prop_2} 
together show 
that $\mathfrak{M}^{\mathfrak{A}}$ reproduces every 
multi-winner voting profile while also admitting 
collective-approval patterns that are not realised in 
$\mathfrak{M}^{\mathfrak{M}}$. 

\subsection{3.3 Argumentative cohesion and $(l, m, n)$-cohesion} 
We now formulate the notions of cohesion in 
$\mathcal{M}^{\mathfrak{A}}$.  
In ordinary approval voting, cohesion is hereditary: if a 
set of voters collectively approves at least $l$ candidates, 
then every non-empty subset also does so. 
Argumentative approvals are non-monotonic, so this hereditary 
property may fail. 
We therefore measure cohesion 
by how widely approval of $l$ candidates persists 
across the set's subsets. 
We make the following observation about $l$-cohesion 
in $(V,C,A,k) \in \mathfrak{M}^{\mathfrak{M}}$: 
if $V' \subseteq V$ is $l$-cohesive, there are at least 
$2^{l \cdot (|V|/k)} - 1$ distinct subsets  
of $V'$ such that, for each $V''$ of them, $|A^{\forall}_{V''}| \geq l$.  
So, the key idea is to shift our focus to sufficiently agreeing subsets. 
With this insight, we define $l$-eligibility. %first. 

\begin{definition}[$l$-eligibility] \rm   
     Let $V'$ be a subset of $V$. $V'$ is {\it $l$-eligible} 
	iff $V'$ is non-empty and $|B_{V'}^{\forall}| \geq l$.  
	By $\eligible{l}{V'}$, we denote the set of all $l$-eligible  
	subsets of $V$', {\it i.e.} 
	$\{V'' \subseteq V' \mid V'' \text{ is $l$-eligible}\}$. 
	 \hfill$\clubsuit$ 
\end{definition} 
{\it In Example 1, $\{v_1\}$ is 1- and 2- eligible---or up to 2-eligible. 
$\{v_1,v_3\}$ is 1-eligible. $\{v_1,v_2,v_3\}$ 
is up to 3-eligible. Generally, $V'$ 
is up to $|B^{\forall}_{V'}|$-eligible. 
}

We now formulate {\it argumentative $l$-cohesion} based on the counts 
of eligible sets: 
\begin{definition}[Argumentative $l$-cohesion]\label{def_argumentative_l_cohesion} \rm 
	Let $V'$ be a subset of $V$, let $\mathcal{P}(V')$ 
	be the set of all subsets of $V'$, and 
	let $l$ be a positive integer.  
	$V'$ is {\it argumentative $l$-cohesive} iff 
	there exists a subset $V''$ of $V'$ such that 
	$V''$ is $l$-eligible and 
	$|\eligible{l}{V''}| \geq 2^{l \cdot |V|/k} - 1$. 
	\hfill$\clubsuit$ 
\end{definition}  
{\it  
  The number of winners ($=k$) is 3 in Example 1. 
  \begin{itemize} 
	  \item Any non-empty subset of $\{v_1,v_2,v_3\}$ is 
		  argumentative 1-cohesive.  
	  \item $\{v_1,v_2\}$, $\{v_2,v_3\}$ and $\{v_1,v_2,v_3\}$ are 
		  argumentative 2-cohesive. The others are not.  
  \end{itemize} 
}
The size requirement that a $l$-cohesive set 
contains at least $l \cdot |V|/k$ voters 
is implicit in Definition \ref{def_argumentative_l_cohesion}. 
\begin{proposition}[Implicit size requirement] \label{prop_implicit_size_requirement} 
	Let $V'$ be a non-empty subset of $V$ and let $l$ be a positive integer. 
	If $|\eligible{l}{V'}| \geq 2^{l \cdot |V|/k}-1$, then necessarily $|V'| \geq l \cdot |V|/k$. 
\end{proposition} 

Monotonicity holds for argumentative $l$-cohesion. 
\begin{proposition}[Monotonicity of argumentative cohesion] \label{prop_monotonicity_argumentative_cohesion} 
    Let $V'$ be a non-empty subset of $V$ and let $l$ be a positive integer. 
	If $V'$ is argumentative $l$-cohesive, 
	then for any superset $V_{super}$ of $V'$, 
	$V_{super}$ is argumentative $l$-cohesive. 
\end{proposition} 
This relaxation (even 
if we apply the same relaxation for $\mathfrak{M}^{\mathfrak{M}}$ and make 
the standard $l$-cohesion similarly monotonic) will turn out to be completely 
harmless for winner selection (Theorem \ref{thm_preservation_m_m_representations}
in Section 4). 

At times, though, stronger notions of cohesion are desirable. 
To this end, we introduce two conditions on cohesion. % each addressing 
The first one requires that at least $m$ ($m \leq l$) 
common candidates 
are approved by all eligible subsets. The second one requires 
that every non-empty subset approves at least $n$ ($n \leq l$) 
candidates. 
\begin{definition}[$(l, m, n)$-cohesion] \rm 
   	Let $V'$ be a subset of $V$, let $l$ be a positive integer 
	and let $m, n$ be non-negative integers 
	not greater than $l$. 
	$V'$ is {\it $(l, m, n)$-cohesive} iff 
	\begin{enumerate} 
		\item   $V'$ is argumentative $l$-cohesive.  
		\item   $m \leq |\bigcap_{V'' \in \eligible{l}{V'}}B^{\forall}_{V''}|$.  
			(\textbf{Common} \textbf{candidates}) 
		\item	$n \leq min_{\emptyset \not= V'' \subseteq V'} |B^{\forall}_{V''}|$.    
			(\textbf{Minimum approval count})  \hfill$\clubsuit$ 
	\end{enumerate} 
\end{definition}  
{\it 
   In Example 1, for instance, 
   \begin{itemize} 
	   \item $\{v_1\}$ is $(1, m_1, n_1)$-cohesive 
   	for all $0 \leq m_1 \leq 1$ and all $0 \leq n_1 \leq 1$.  
\item $\{v_2\}$ is  $1.m_2.n_2$-cohesive 
   for all $0 \leq m_2 \leq 1$ and all $0 \leq n_2 \leq 1$. 
\item $\{v_1,v_2\}$ is $(l_{12}, m_{12}, n_{12})$-cohesive for all 
		   $1 \leq l_{12} \leq 2$, all $0 \leq m_{12} \leq 1$ and all 
		   $0 \leq n_{12} \leq l_{12}$.   
	   \item $\{v_1,v_3\}$ is $(1, m_{13}, n_{13})$-cohesive for all 
		   $0 \leq m_{13} \leq 1$ and 
		   all $0 \leq n_{13} \leq 1$. 
	   \item $\{v_1,v_2,v_3\}$ is 
		   $(l_{123}, m_{123}, n_{123})$-cohesive 
		   for all $1 \leq l_{123} \leq 2$, all 
		   $0 \leq m_{123} \leq 1$ and all $0 \leq n_{123} \leq 1$. 
   \end{itemize} 
}

\begin{proposition}[Conservation] \label{prop_conservation} 
	Let $V'$ be a subset of $V$ and let $l$ be a positive integer. 
	$V'$ is argumentative $l$-cohesive iff $V'$ is $(l, 0, 0)$-cohesive.  
\end{proposition} 

\begin{proposition}[Cohesion monotonicity] \label{prop_cohesion_monotonicity} 
     Let $V'$ be a subset of $V$, let $l$ be a positive integer, 
	and let $m, n$ be 
	non-negative 
	integers. Let $m'$ be a non-negative integer not 
	greater than $m$, 
	and let $n'$ be a non-negative integer not greater than $n$. 
	If $V'$ is $(l, m, n)$-cohesive, 
	then $V'$ is $(l, m', n')$-cohesive. 	
\end{proposition} 

In a limited case, $(l, m, n)$-cohesion can be lifted to 
$(l, (m+1), n)$-cohesion or $(l, m, (n+1))$-cohesion. 

\begin{proposition}[Lifting]  \label{prop_lifting} 
       Let $V'$ be a subset of $V$, let $l$ be  
	a positive integer and let $n$ be a non-negative integer. 
	If $V'$ is $(l, 0, n)$-cohesive and $l$-eligible, then 
	$V'$ is $(l, 1, n)$-cohesive.   
\end{proposition}

\section{4 Justified Representation} 
\label{sec_justified_representation_axioms_and_consequences}    
\subsection{4.1 Justified representation axioms} 
We now introduce justified representation axioms for $\mathfrak{M}^{\mathfrak{A}}$. 
The main difficulty is as follows.
On the one hand, for a multi-winner voting profile $(V,C,A,k)$, 
$A^{\forall}_{V'} = \bigcap_{v \in V'} A^{\forall}_{\{v\}}$, 
so any candidate approved by a set of voters is also approved individually; 
on the other hand, for an argumentative voting profile 
$(V,C,B,k)$, it can be that $B^{\forall}_{V'} \not= \bigcap_{v \in V'} 
B^{\forall}_{\{v\}}$. Example 1 is a concrete such case. 
Some candidate approved by a set of voters may only be 
approved by the particular set and not by its strict \mbox{super-/sub-sets.}% or strict subsets.  

Recall $\mathsf{JR}$, $\mathsf{PJR}$ and $\mathsf{EJR}$ from Section 2, 
even though they are meant to satisfy the demands of sets of voters, 
any winner set $W$ providing $\mathsf{JR}/\mathsf{PJR}/\mathsf{EJR}$ is necessarily a 
subset of $\bigcup_{v \in V} A_{v}$;
hence $W$ is identifiable by comparing it against individuals' approvals alone.  
The above-described difference suggests that, with an argumentative voting 
profile, we may genuinely have to 
select winners from among collectively approved candidates.

To account for this distinction, we anchor $W$'s provision 
to approvals of eligible sets, 
which gives us: 

\begin{definition}[Justified representation axioms] \rm    
	Let $W$ be a size-$k$ subset of $C$. 
	Let $h: \mathbb{N} \rightarrow \mathbb{N} \times \mathbb{N}$ 
	be such that: $h(l) = (h_1(l), h_2(l))$; 
	$0 < l$; 
	$0 \leq h_1(l) \leq l$; and $0 \leq h_2(l) \leq l$. 
	{\it With respect to $h$, $W$ provides} 
	\begin{itemize} 
		\item $\mathsf{ArgJR}$ iff, for 
			any $(1, h_1(1), h_2(1))$-cohesive 
			subset $V'$ of $V$, there is some 
			 $V''\in \eligible{1}{V'}$ 
			such that 
			$|W \cap B^{\forall}_{V''}| \geq 1$. 
		\item $\mathsf{ArgPJR}$ iff, for 
			any $(l, h_1(l), h_2(l))$-cohesive 
			subset $V'$ of $V$,    
			it holds that $ 
			|W \cap (\bigcup_{V'' \in \eligible{l}{V'}}B_{V''}^{\forall})|
			\geq l$. 
		\item $\mathsf{ArgEJR}$ iff, for any 
			$(l, h_1(l), h_2(l))$-cohesive 
			subset $V'$ of $V$,  
			there is some $v \in V'$ such that 
			$|W \cap (\bigcup_{V'' \in \eligible{l}{V'}, v \in V''}B_{V''}^{\forall}) | \geq l$. 
		\item $\mathsf{ArgEJR}\text{-}\mathsf{Spot}$ iff, for any 
			$(l, h_1(l), h_2(l))$-cohesive 
			subset $V'$ of $V$,  
			there is some $V'' \in \eligible{l}{V'}$ 
			such that 
			$|W \cap B_{V''}^{\forall} 
			| \geq l$.  \hfill$\clubsuit$ 
	\end{itemize} 
\end{definition} 
If $V'$ is $(l, h_1(l), h_2(l))$-cohesive, 
\begin{itemize} 
	\item For $\mathsf{ArgPJR}$: for each of $l$ winners $c_1, \ldots, c_l \in W$, 
		it has to be approved by some $l$-eligible subset of $V'$. 
	\item For $\mathsf{ArgEJR}$: 
		there has to be some individual $v \in V'$ such that, 
		for each of $l$ winners $c_1, \ldots, c_l \in W$, 
		it is approved by some $l$-eligible subset of $V'$ that contains
		$v$. 
\end{itemize}

Clearly, they do not just compare $W$ against approvals at individual level. 
It is, however, possible to require that $W$ be compared against 
a single eligible set. Hence, we have a stronger axiom 
$\mathsf{ArgEJR}\text{-}\mathsf{Spot}$ of $\mathsf{ArgEJR}$. 
$\mathsf{ArgJR}$ and $\mathsf{ArgPJR}$ can be equally strengthened 
into $\mathsf{ArgJR}\text{-}\mathsf{Spot}$ and $\mathsf{ArgPJR}\text{-}\mathsf{Spot}$, 
but $\mathsf{ArgJR}\text{-}\mathsf{Spot}$ is $\mathsf{ArgJR}$ itself, 
and $\mathsf{ArgPJR}\text{-}\mathsf{Spot}$ collapses onto 
$\mathsf{ArgEJR}\text{-}\mathsf{Spot}$. 

{\it In Example 1, $k=3$, so there are twenty possibilities for $W$. 
With respect to $h$ where each of $h(1)$ and $h(2)$ is either 
$(0,0)$ or $(1,1)$, 
    \begin{itemize} 
	    \item Ten of them provide $\mathsf{ArgEJR}\text{-}\mathsf{Spot}$. 
		    \begin{itemize} 
			    \item Any $W$ with $c_1 \in W$ and 
				    $W \not= \{c_1,c_3,c_6\}$, 
				   plus $\{c_2,c_4,c_5\}$ provide it. 
		    \end{itemize} 
	    \item Eleven of them provide $\mathsf{ArgEJR}$. 
		    \begin{itemize} 
			    \item $\{c_4,c_5,c_6\}$ provides 
				    $\mathsf{ArgEJR}$ but not 
				    $\mathsf{ArgEJR}\text{-}\mathsf{Spot}$. 
		    \end{itemize} 
	    \item Fourteen of them provide $\mathsf{ArgPJR}$. 
		    \begin{itemize} 
			    \item $\{c_2,c_3,c_5\}$, 
				    $\{c_2,c_5,c_6\}$ 
				    and $\{c_3,c_4,c_5\}$ provide 
				    $\mathsf{ArgPJR}$ but not 
				    $\mathsf{ArgEJR}$. 
		    \end{itemize} 
	    \item Fifteen of them provide $\mathsf{ArgJR}$. 
		    \begin{itemize} 
			    \item $\{c_1,c_3,c_6\}$ provides 
				    $\mathsf{ArgJR}$ but not 
				    $\mathsf{ArgPJR}$. 
		    \end{itemize} 
    \end{itemize} 
With respect to $h$ where $h(l) = (l,l)$ ($l \in \{1,2\}$), 
$\{v_1,v_2\}$ 
ceases to be 
$(2, h_1(2), h_2(2))$-cohesive. 
Provision 
of $\mathsf{ArgJR}$ remains unchanged. For the others:  
    \begin{itemize} 
	    \item Thirteen of them provide $\mathsf{ArgEJR}\text{-}\mathsf{Spot}$.
		    \begin{itemize} 
			    \item $\{c_1, c_3, c_6\}$, 
				    $\{c_2,c_3,c_5\}$ and $\{c_2,c_5,c_6\}$ 
				    additionally provide 
				    $\mathsf{ArgEJR}\text{-}\mathsf{Spot}$. 
		    \end{itemize} 
	    \item Fourteen of them provide $\mathsf{ArgEJR}$. 
		    \begin{itemize} 
			    \item $\{c_4,c_5,c_6\}$ provides 
				    $\mathsf{ArgEJR}$ but not 
				    $\mathsf{ArgEJR}\text{-}\mathsf{Spot}$. 
		    \end{itemize}  
	    \item Fifteen of them provide $\mathsf{ArgPJR}$. 
		    \begin{itemize} 
			    \item $\{c_3, c_4, c_5\}$ provides 
				    $\mathsf{ArgPJR}$ but not 
				    $\mathsf{ArgEJR}$. 
		    \end{itemize} 
    \end{itemize} 

}
As this illustration shows, greater values for $h_1(l)$ and $h_2(l)$  
lead to fewer, or at best the same number of, requirements 
for a size-$k$ subset of $C$ to satisfy in order to 
provide 
the axioms. As such, two particular $h$ 
make canonical cases: the {\it permissive} $h$ 
is such that 
$h(l) = (0,0)$ for all $l$; and the {\it robust} 
$h$ is such that $h(l) = (l,l)$ for all $l$.

In the remainder, we establish: correspondence with 
$\mathsf{JR}$, $\mathsf{PJR}$ and $\mathsf{EJR}$ (Section 4.2); 
unconditional provision of $\mathsf{ArgJR}$ and provision under 
robust $h$ for the stronger axioms (Section 4.3); 
and polynomial-time construction but $\mathsf{coNP}$-hard verification for $\mathsf{ArgJR}$ (Section 4.4).

\subsection{4.2 Representation correspondence results}
Proposition \ref{prop_approval_preserving_transformation} gave 
an approval-preserving embedding of multi-winner voting profiles 
into argumentative ones. 

Let $\nu$ be a member of $\{ArgJR^h, ArgPJR^h, ArgEJR^h,\linebreak ArgEJRS^h\}$, 
and let $\nu(V, C, B, k)$ be the set of all size-$k$ subsets $W$ of $C$ such that $W$ provides: 
$\mathsf{ArgJR}$ with respect to $h$ if $\nu$ is $ArgJR^h$; 
$\mathsf{ArgPJR}$ with respect to $h$ if $\nu$ is $ArgPJR^h$; 
$\mathsf{ArgEJR}$ with respect to $h$ if $\nu$ is $ArgEJR^h$; and 
$\mathsf{ArgEJR}\text{-}\mathsf{Spot}$ with respect to $h$ if $\nu$ is 
$ArgEJRS^h$. Then, the following correspondence results hold. 
\begin{theorem}[Preservation of $\mathfrak{M}^{\mathfrak{M}}$ representations]\label{thm_preservation_m_m_representations} 
	Let $x \equiv (V, C, A, k)$ be a member of 
	$\mathfrak{M}^{\mathfrak{M}}$. 
	For any size-$k$ subset $W$ of $C$  and 
	any $h$, 
	\begin{itemize} 
		\item $W \in JR(x)$ iff $W \in ArgJR^h(\tau(x))$. 
		\item  $W \in PJR(x)$ iff $W \in ArgPJR^h(\tau(x))$. 
		\item  $W \in EJR(x)$ iff $W \in ArgEJR^h(\tau(x))$\\ iff 
		$W \in ArgEJRS^h(\tau(x))$. 
	\end{itemize}  
\end{theorem} 

The set inclusions  
among justified representation axioms in $\mathfrak{M}^{\mathfrak{A}}$ are as expected. 
\begin{theorem}[Inclusion hierarchy] \label{thm_set_inclusion_representations} 
	Let $x \equiv (V,C,B,k)$ be a member of $\mathfrak{M}^{\mathfrak{A}}$, 
	then, with respect to every $h$, $ArgEJRS^h(x) \subseteq ArgEJR^h(x) \subseteq 
	ArgPJR^h(x) \subseteq ArgJR^h(x)$. 
\end{theorem}  

Also, they are separable.  
\begin{theorem}[Separations] \label{thm_separations} 
	For each of the following statements, 
	there exist an argumentative voting profile 
	$x \equiv (V, C, B, k)$, 
	 a size-$k$ subset $W$ of  $C$ and an $h$  
	 for which the statement holds. 
	\begin{itemize} 
		\item $W \in ArgEJR^h(x)$ and $W \not\in ArgEJRS^h(x)$. 
		\item $W \in ArgPJR^h(x)$ and $W \not\in ArgEJR^h(x)$. 
		\item $W \in ArgJR^h(x)$ and $W \not\in ArgPJR^h(x)$. 
	\end{itemize} 
\end{theorem}

\subsection{4.3 Existence results}    
Provision of only $\mathsf{ArgJR}$ is unconditionally guaranteed. 
\begin{theorem}[Existence for $\mathsf{ArgJR}$] \label{thm_existence_ajr} 
	For any $(V, C, B, k) \in \mathfrak{M}^{\mathfrak{A}}$,  with respect to any $h$, 
      there is some size-$k$ subset $W$ of $C$ such that 
	$W$ provides $\mathsf{ArgJR}$. 
\end{theorem} 

\begin{theorem}[Impossibility for {\small $\mathsf{ArgPJR}/\mathsf{ArgEJR}/\mathsf{ArgEJR}\text{-}\mathsf{Spot}$}] \label{thm_impossibility_zero} 
	For any $\rho \in \{\mathsf{ArgPJR}, \mathsf{ArgEJR}, \mathsf{ArgEJR}\text{-}\mathsf{Spot}\}$,
	there is some 
	$(V, C, B, k) \in \mathfrak{M}^{\mathfrak{A}}$  
	and some $h$ 
	such that, with respect to $h$, 
      no size-$k$ subset $W$ of $C$ provides 
	 $\rho$. 
\end{theorem} 

Nonetheless, the robust $h$ gives us the following result. 

\begin{theorem}[Existence for all] \label{thm_existence_el} 
	For any $\rho \in \{\mathsf{ArgPJR},\\\mathsf{ArgEJR},\mathsf{ArgEJR}\text{-}\mathsf{Spot}\}$ and 
	for any argumentative voting profile 
	$(V, C, B, k) \in \mathfrak{M}^{\mathfrak{A}}$ 
      there is some size-$k$ subset $W$ of $C$ such that 
	$W$ provides $\rho$ with respect to the robust $h$. 
\end{theorem}  

\subsection{4.4 Computational complexity results}  
\newcommand{\ArgEJRSpot}{\mathsf{ArgEJR}$-$\mathsf{Spot}} 
\newcommand{\ArgEJR}{\mathsf{ArgEJR}} 
\newcommand{\ArgPJR}{\mathsf{ArgPJR}} 
\newcommand{\ArgJR}{\mathsf{ArgJR}} 
Verification is $\mathsf{coNP}$-hard  
for $\mathsf{ArgJR}$ and thus for all axioms. 
\begin{theorem}[Verification hardness] \label{thm_verification_hardness} 
	With respect to any $h$, determining if a given 
       size-$k$ subset $W$ of $C$ 
	provides $\mathsf{ArgJR}$, 
	$\mathsf{ArgPJR}$, $\mathsf{ArgEJR}$ 
	or $\mathsf{ArgEJR}$-$\mathsf{Spot}$ is 
	$\mathsf{coNP}$-hard. 
\end{theorem}   

\begin{algorithm}[tb]
	\caption{\textsf{GreedyGrounded} algorithm}
\label{alg:algorithm}
	\textbf{Input}: An argumentative voting profile $(V,C,B,k)$\\
%	\textbf{Temporaries}: Two maps $Alloc$ and $D$. 
%	 A subset $\Gamma$ of $\mathcal{P}(V)$. 
	\textbf{Output}: A size-$k$ subset of $C$. 
\begin{algorithmic}[1] %[1] enables line numbers
\STATE $out \leftarrow \emptyset$. $Alloc \leftarrow \{(v, 0) \mid v \in V\}$.
 $D \leftarrow \emptyset$. 
	  $waiting \leftarrow 
		\mathsf{notAttkd}(C,  
		\{V' \subseteq V \mid  
	\text{for some } c \in C, V' = (V \backslash   
	\{v \in V \mid \exists c'.(c',c) \in R_v\}) 
	V' \text{ is a maximal set satisfying: } \forall c' \in C.(c',c) \not\in \bigcup_{v \in V'}R_v
	\text{ and } 
		|V'| \geq |V|/k\}
		)$. {\color{gray}{// initialise.}} 

\WHILE{true}  
	\WHILE{$waiting \not= \emptyset$} 
	\STATE $(c, V') \leftarrow$ a greatest member of $waiting$. $waiting \leftarrow waiting \backslash \{(c, V')\}$. 
		$out \leftarrow out \cup \{c\}$.   for each $v \in V'$, $Alloc[v] \leftarrow Alloc[v]+1$. 
	\IF {$|out| = k$} 
	\STATE \textbf{return} $out$.  
	\ENDIF 
	 \IF {$V'$ is not a key in $D$} 
	\STATE $D \leftarrow D \cup \{(V',\{c\})\}$.  
	\ELSE 
	\STATE $D[V'] \leftarrow D[V'] \cup \{c\}$. 
	\ENDIF  
	\ENDWHILE 
	\STATE  
	$waiting \leftarrow \bigcup_{V' \in keys(D)} \mathsf{notAttkd}(C \backslash (D[V'] \cup 
	\{c \in C \mid \exists c' \in D[V'].(c',c) \in \bigcup_{v\in V'}R_v\}),\{V'\})$. 
	\IF {$waiting = \emptyset$}  
	     \STATE \textbf{return} a size-$k$ subset of $C$ that contains $out$. 
	\ENDIF 
\ENDWHILE
\end{algorithmic}
\end{algorithm}
Nevertheless, a size-$k$ winner set providing $\mathsf{ArgJR}$  
with respect to any $h$ can be constructed in polynomial time. 
$\mathsf{GreedyGrounded}$ in Algorithm 1 takes $(V,C,B,k)$ as the input 
and outputs 
a size-$k$ subset of $C$. 

There are three auxiliary data structures.  
\begin{enumerate} 
	\item $Alloc: V \rightarrow \mathbb{N}$ 
records an allocation score of each voter.  
For $v \in V$, the greater its value $Alloc[v]$ is, 
the more allocated $v$ is considered to be.   
\item  $D$ is a partial map 
	from voter subsets ($\subseteq \mathcal{P}(V)$) to 
		candidate subsets ($\subseteq \mathcal{P}(C)$). 
It records the candidates selected on 
		behalf of each stored subset of voters. 
%set(s) of voters. 
		$keys(D)$ denotes the set of keys in $D$, and 
$D[V']$ denotes the value of $V' \in keys(D)$.  
\item 
The priority set $waiting$ contains  
pairs of $C \times \mathcal{P}(V)$ ranked in  
the following order: $(c_1, V'_1) \geq (c_2, V'_2)$ (with non-empty $V'_1$ and $V'_2$) iff   
there is some integer $0 \leq x$ such that, for every integer 
$0 \leq y < x$,  
the same number of $y$ occurs in the multisets 
$\uplus_{v \in V'_1}\{Alloc[v]\}$ and $\uplus_{v \in V'_2}\{Alloc[v]\}$, 
and the multiplicity of $x$ is greater in 
$\uplus_{v \in V'_1}\{Alloc[v]\}$. Thus, 
	pairs representing more voters with low allocation 
		scores receive higher priority. 
\end{enumerate} 

The $waiting$ set is generated by a function 
$\mathsf{notAttkd}: \mathcal{P}(C) \times \mathcal{P}(\mathcal{P}(V)) 
\rightarrow \mathcal{P}(C \times \mathcal{P}(V))$, 
which is such that $\mathsf{notAttkd}(C', \{V'_1, \ldots, V'_j\}) =
\{(c, V') \mid c \in C' \text{ and } \exists 1 \leq i \leq j.(V' = V'_i \text{ and } 
\forall c' \in C'.(c', c) \not\in ((C' \times C') \cap (\bigcup_{v \in V'_i} R_v)))\}$.  
For example, $\mathsf{notAttkd}(C', \{V'_1, V'_2\})$ 
contains every pair $(c_i, V_i')$ ($i \in \{1,2\}$) 
such that $c_i \in C'$ and $V'_i$ does 
not attack $c_i$ in the graph $(C', (C'\times C') \cap \bigcup_{v \in V'_i} R_v)$. 

With these, 
$\mathsf{GreedyGrounded}$ constructs its output.  

For each $c \in C$, 
if $V' \subseteq V$ does not attack $c$ in $(C, \bigcup_{v \in V'} R_v)$, $|V'| \geq |V|/k$, and is a maximal such set, then 
$(c, V')$ is added into $waiting$ (line 1).  
The algorithm repeatedly removes a highest-priority pair 
$(c_x, V'_x)$ from $waiting$, $c_x$ is added to $out$, 
and $Alloc$ is updated for members of $V'$ (line 4). 
$D$ is also updated (lines 8$\sim$12).  
If the initial $waiting$ set yields fewer than $k$ winners, 
the algorithm tries to favour 
members of $keys(D)$ with more winners.
So, it generates a new $waiting$ set comprising 
all $(c_x, V'_x)$ where $c_x \in (C \backslash D[V_x'])$, 
$V'_x \in keys(D)$ and 
$D[V'_x]$ defends $c_x$ in $(C, \bigcup_{v \in V'_x} R_v)$ (line 14).\footnote{Algorithm 1 achieves this with 
$\mathsf{notAttkd}$.} 
Selection continues (line 4) 
until $k$ distinct winners are selected (lines 5$\sim$7) 
or a new $waiting$ set is empty (line 15) 
in which case the remaining winners are arbitrary (line 16).   

The invariants that $D[V'] \subseteq B^{\forall}_{V'}$ for every $V' \in keys(D)$ and that $keys(D) < k$ hold at every 
execution of line 3. 

\begin{theorem}[Polynomial-time construction]\label{thm_polynomial_time_construction} {\ } For any $(V,C,B,k) 
	\in \mathfrak{M}^{\mathfrak{A}}$ 
	and any $h$, 
	{\small \textsf{GreedyGrounded}} computes in polynomial time 
	a 
	\mbox{size-$k$} subset $W$ of $C$ 
	providing $\mathsf{ArgJR}$ with respect to $h$.    
	With adjacency-matrix representations, 
	its running time is $O(k^2 \cdot |V| \cdot |C|^2)$. 
	
	The running time with an incremental 
	implementation of: the union graphs; 
	grounded-extension computations 
	for each fixed $V' \in keys(D)$; 
	and priority updates, 
	is $O(k\cdot |V|\cdot |C|^2)$. % time and uses $O(k\cdot |C|^2 
%	+ |V| \cdot |C|^2)$ space. 
\end{theorem}

\section{5 Conclusions}      
We introduced multi-winner voting with argumentative ballots, 
where approvals are induced by grounded acceptance. 
The framework conservatively generalises approval voting while 
retaining relational information that ordinary approval ballots
do not capture. 

Individual approvals may be fixed, 
but combining ballots can lead to different collective approvals 
through attack and defence.
This non-monotonicity demands that justified representation be 
reformulated. 

Our argumentative cohesion notions and justified 
representation axioms 
$\mathsf{ArgJR}$, $\mathsf{ArgPJR}$, $\mathsf{ArgEJR}$ and 
$\mathsf{ArgEJR}$-$\mathsf{Spot}$ 
conservatively generalise $\mathsf{JR}$, $\mathsf{PJR}$ and 
$\mathsf{EJR}$, and preserve their inclusion hierarchy. 
We showed that a winner set providing $\mathsf{ArgJR}$ always 
exists, whereas winner sets providing the stronger axioms 
need not exist. The robust $h$ restores existence for all of them. 
Identifying a more precise boundary remains open. 
Computationally, a winner set satisfying $\mathsf{ArgJR}$ is 
polynomial-time constructible, even though verifying whether 
a given winner set provides it is $\mathsf{coNP}$-hard.

Future work includes adapting further representation axioms \cite{Kalayci25,Brill23}, incorporating candidate costs \cite{Peters21}, 
understanding strategic behaviour, and developing practical methods for eliciting argumentative ballots.

\section*{Author Contributions} 
R.A.: Conceptualisation, methodology, formal analysis, investigation, software, validation and writing - original draft. 

\noindent 
H.O.: Formal analysis of the proof of Theorem \ref{thm_verification_hardness}. 

\section*{Acknowledgements}  
This work was supported in part by JSPS KAKENHI Grant Number 25K15245.

\bibliography{references} 

\section*{Appendix: Proofs} 
This appendix contains all proofs of the claims in the 
main text. To prove Proposition \ref{prop_2}, we use 
	Lemma \ref{lem_collective_approval_gadgets}. 
	To prove Theorem \ref{thm_preservation_m_m_representations}, 
	we use Lemma 
	\ref{lem_relationship_between_cohesion_argumentative_cohesion}. To prove Theorem \ref{thm_impossibility_zero}, we use 
	Lemmas \ref{lem_collective_approval_gadgets} and 
	\ref{lem_non_approval_gadgets}. 
	These lemmas appear in the final section. 
	The existence proofs for Theorem \ref{thm_existence_ajr} 
	and in particular Theorem 
	\ref{thm_existence_el},  
	use a direct inductive covering argument, 
	rather than optimisation of a specific voting rule. 
	To the best of our knowledge,
	this proof strategy is novel in the study 
	of justified representation.  
	All proofs are certified in Lean 4.

\subsection*{A1 Proofs of Propositions}   
\noindent 
\textbf{Proof of Proposition \ref{prop_approval_preserving_transformation}}. 
Let $\kappa: \mathfrak{M}^{\mathfrak{M}} \rightarrow 
   \mathfrak{M}^{\mathfrak{A}}$ be a function with the same 
   domain and co-domain as $\tau$.   
   If $\kappa$ satisfies the following condition, 
   $\kappa$ is an instantiation of $\tau$:  
   for any $(V,C,A,k) \in \mathfrak{M}^{\mathfrak{M}}$, 
	let $(V,C,B,k)$ denote $\kappa(V,C,A,k)$, 
    for any $v \in V$ and any $c \in C$,  
	\begin{itemize} 
	   \item if $c \in A_v$, then there is no $c' \in C$ 
		   such that $(c',c)$ is an edge in $B_{\{v\}}$.  
	   \item for any $(C, R_v')$ where 
		   $R_v'$ includes all edges in $B_{\{v\}}$, 
		   if $c \not\in A_v$, then 
		   $c$ is not in the grounded extension of $(C,R_v')$. 
   \end{itemize}  
   One such $\kappa$ is: 
   for any $x \in \mathfrak{M}^{\mathfrak{M}}$,   
   $\kappa(x) \equiv (V,C,B,k)$ comes with $B = (B_{\{v_1\}}, \ldots, B_{\{v_{|V|}\}})$ 
   where 
$B_{\{v_i\}}$ is $(C, \{(c, c) \mid c \not\in A_{v_i}\})$. \hfill$\Box$ \\

\noindent 
\textbf{Proof of Proposition \ref{prop_2}}. 
	Use collective approval gadgets; Lemma \ref{lem_collective_approval_gadgets}. Obvious, since 
	Fact 1 holds for $\tau'(V,C,B,k)$ but does not need hold 
	for $(V,C,B,k)$. \hfill$\Box$\\

\noindent 
\textbf{Proof of Proposition \ref{prop_implicit_size_requirement}}. 
If $|V'| < l \cdot |V|/k$, the number of subsets of $V'$ 
	is strictly smaller than $2^{l \cdot |V|/k}$. Even if 
	every non-empty subset is $l$-eligible, $|\eligible{l}{V'}|$ is strictly 
	smaller than $2^{l \cdot |V|/k} - 1$. 
\hfill$\Box$\\ 

\noindent 
\textbf{Proof of Proposition \ref{prop_monotonicity_argumentative_cohesion}}. 
$V_{super}$ contains all $l$-eligible subsets of $V'$.
\hfill$\Box$\\

\noindent 
\textbf{Proof of Proposition \ref{prop_conservation}}.  
\textbf{If}: Vacuous.  
	\textbf{Only if}: The first condition holds by assumption. The second condition 
	holds because $|\bigcap_{V'' \in \eligible{l}{V'}}B^{\forall}_{V''}|$ is at least 0. 
	The third condition holds because $min_{V'' \subseteq V'} |B^{\forall}_{V''}|$ 
	is at least 0.
\hfill$\Box$\\ 

\noindent 
\textbf{Proof of Proposition \ref{prop_cohesion_monotonicity}}. 
If $V'$ is $(l, m, n)$-cohesive, then: 
	\begin{itemize} 
		\item Every member of $\eligible{l}{V'}$  
	approves at least $m$ same candidates, which means 
	at least $m'$ same candidates are approved.  
		\item Every non-empty subset of $V'$ approves 
			at least $n$ candidates, which means 
			at least $n'$ candidates are approved. 
	\end{itemize} 
\hfill$\Box$\\

\noindent 
\textbf{Proof of Proposition \ref{prop_lifting}}. 
By assumption, $V'$ is $l$-eligible. So, 
	$B^{\forall}_{V'} \not= \emptyset$. 
	Thus, there must exist a candidate $c$ 
	in $\bigcap_{v \in V'} B^{\forall}_{\{v\}}$ 
	which is not attacked in $B_{V'}$. 
	For, otherwise, $B_{V'}$ would contain no unattacked candidate 
	and $B^{\forall}_{V'}$ would be an empty set. 
	Hence, for any subset $V''$ of $V'$, 
	$c \in B^{\forall}_{V''}$.
\hfill$\Box$\\ 

\subsection*{A2 Proofs of Theorems} 
\textbf{Proof of Theorem \ref{thm_preservation_m_m_representations}}. 
Let $(V,C,B,k)$ denote $\tau(x)$. 
	By Proposition \ref{prop_approval_preserving_transformation},  
	for any non-empty subset $V'$ of $V$, 
	$(\bigcap_{v \in V'} A_v) = A^{\forall}_{V'} = B^{\forall}_{V'}$.  

	For the first obligation, \textbf{Only if}:   
	Let $\Gamma(1)$ denote the set of all non-empty subsets $V'$ of 
	$V$ such that $V'$ is $1$-cohesive in $x$ and that 
	there is no strict non-empty subset $V''$ of $V'$ 
	such that $V''$ is $1$-cohesive in $x$.  
	Since $W$ provides $\mathsf{JR}$, 
	for each $V' \in \Gamma(1)$, 
	there is some $v \in V'$ such that $|W \cap A^{\forall}_{\{v\}}| \geq 1$.  
	For any $1$-cohesive subset $V'_o$ in $x$ not in $\Gamma(1)$, 
	there is some $V' \in \Gamma(1)$ such that $V' \subset V'_o$. 
	
	By Lemma 
	\ref{lem_relationship_between_cohesion_argumentative_cohesion}, 
	$V'$ is $(1, 1, 1)$-cohesive in $\tau(x)$ and that 
	there is no strict non-empty subset $V''$ of $V'$ such that 
	$V''$ is argumentative $1$-cohesive in $\tau(x)$. 
	Since $\{v\} \in \eligible{l}{V'}$, 
	$|W \cap B^{\forall}_{\{v\}}| \geq 1$.   
	By Proposition \ref{prop_cohesion_monotonicity}, 
	$V'$ is $(1, h_1(1), h_2(1))$-cohesive in $\tau(x)$.  
	By Proposition \ref{prop_monotonicity_argumentative_cohesion}, 
	for any superset $V_{super}$ of $V'$,  $V_{super}$ is 
	argumentative 
	$1$-cohesive. If $V_{super}$ is $(1, h_1(1), h_2(1))$-cohesive, 
	$\{v\} \in \eligible{1}{V_{super}}$; otherwise, 
	there is nothing to show.

	\textbf{If}: 
	Let $\Delta(h, 1)$ denote the set of 
	all non-empty subsets $V'$ of $V$ such that $V'$ is $(1, h_1(1), h_2(1))$-cohesive 
	in $\tau(x)$ and that there is no strict non-empty subset $V''$ 
	of $V'$ such that $V''$ is $(1, h_1(1), h_2(1))$-cohesive 
	in $\tau(x)$. 
	Since $W$ provides $\mathsf{ArgJR}$ with respect to $h$,  
	for each $V' \in \Delta(h,1)$, there is some 
	$V'' \in 
	\eligible{1}{V'}$ such that $|W \cap B^{\forall}_{V''}| 
	\geq 1$. 
	Since, by Fact 1 and Proposition 1, $V'$ is $(1, 1, 1)$-cohesive in $\tau(x)$, 
	there is in particular some $v \in V''$ such 
	that $|W \cap B^{\forall}_{\{v\}}| 
	\geq 1$. For any $(1, h_1(1), h_2(1))$-cohesive subset 
	$V'_o$ 
	in $\tau(x)$ not in $\Delta(h,1)$, 
	there is some $V' \in \Delta(h,1)$ such that 
	$V' \subset V'_o$ 
	and that $v \in \eligible{1}{V'_o}$ for each $v \in V'$. 
	 
	By Lemma \ref{lem_relationship_between_cohesion_argumentative_cohesion}, 
	$V'$ is $1$-cohesive in $x$ and that there is no strict 
	non-empty 
	subset $V''$ of $V'$ such that $V''$ is $1$-cohesive in $x$. 
	Trivially, $|W \cap A^{\forall}_{\{v\}}| \geq 1$.   
	For any superset $V_{super}$ of $V'$,  
	$v \in V_{super}$. 

	For the second obligation, \textbf{Only if}:  
	For each $1 \leq l \leq k$, 
	let $\Gamma(l)$ denote the set of all non-empty subsets $V'$ of 
	$V$ such that $V'$ is $l$-cohesive in $x$ and that 
	there is no strict non-empty subset $V''$ of $V'$ 
	such that $V''$ is $l$-cohesive in $x$.  
	Since $W$ provides $\mathsf{PJR}$, 
	for each $1 \leq l \leq k$ 
	and each $V' \in \Gamma(l)$, 
	$|W \cap \bigcup_{v \in V'} A^{\forall}_{\{v\}}| \geq l$.  
	For any $l$-cohesive subset $V'_o$ in $x$ not in $\Gamma(l)$, 
	there is some $V' \in \Gamma(l)$ such that $V' \subset V'_o$.

	By Lemma \ref{lem_relationship_between_cohesion_argumentative_cohesion}, 
	$V'$ is $(l, l, l)$-cohesive in $\tau(x)$ and that 
	there is no strict non-empty subset $V''$ of $V'$ such that 
	$V''$ is argumentative $l$-cohesive in $\tau(x)$.   
	Since $\{v\} \in \eligible{l}{V'}$ for each $v \in V'$,  
	$|W \cap \bigcup_{v \in V'} B^{\forall}_{\{v\}}| \geq l$.  
	Then, clearly, $|W \cap (\bigcup_{V'' \in \eligible{l}{V'}}B_{V''}^{\forall})| 
	\geq l$. For any superset $V_{super}$ of $V'$,  
	if $V_{super}$ is $(l, h_1(l), h_2(l))$-cohesive in 
	$\tau(x)$,  
	then for each $v \in V'$, $\{v\} \in \eligible{l}{V_{super}}$; otherwise, 
	there is nothing to show. 

	\textbf{If}: 
	For each $1 \leq l \leq k$, let $\Delta(h, l)$ denote the set of 
	all non-empty subsets $V'$ of $V$ such that $V'$ is $(l, h_1(l), h_2(l))$-cohesive 
	in $\tau(x)$ and that there is no strict non-empty subset $V''$ 
	of $V'$ such that $V''$ is $(l, h_1(l), h_2(l))$-cohesive 
	in $\tau(x)$. 
	Since $W$ provides $\mathsf{ArgPJR}$ with respect to $h$,  
	for each $1 \leq l \leq k$ 
	and each $V' \in \Delta(h,l)$, there is some $V'' \in 
	\eligible{l}{V'}$ such that 
	$|W \cap 
	(\bigcup_{V'' \in \eligible{l}{V'}} B^{\forall}_{V''})| 
	\geq l$.
	Since, by Fact 1 and Proposition 1, $V'$ is $(l, l, l)$-cohesive in $\tau(x)$,  
	$\bigcup_{V'' \in \eligible{l}{V'}} B^{\forall}_{V''} 
	\subseteq \bigcup_{v \in V'}B^{\forall}_{\{v\}}$, so, 
	in particular, we have $|W \cap 
	(\bigcup_{v \in V'} B^{\forall}_{\{v\}})| 
	\geq l$. Hence, $|W \cap \bigcup_{v \in V'} A^{\forall}_{\{v\}}| \geq l$. 
	For any $(l, h_1(l), h_2(l))$-cohesive subset $V'_o$ 
	in $\tau(x)$ not in $\Delta(h,l)$, 
	there is some $V' \in \Delta(h,l)$ such that $V' \subset V'_o$ 
	and that $v \in \eligible{l}{V'_o}$ for each $v \in V'$. % $v \in V'_o$. 

	For the third obligation, 
	firstly for: $W \in ArgEJR^h(\tau(x))$ iff  
	$W \in ArgEJRS^h(\tau(x))$,  
	\textbf{Only if}:  
	Since $W$ provides $\mathsf{ArgEJR}$ with respect to $h$, 
	for any $(l, h_1(l), h_2(l))$-cohesive subset $V'$ of $V$,  
	there is some $v \in V'$ such that 
	$|W \cap (\bigcup_{V'' \in \eligible{l}{V'}, v \in V''} B^{\forall}_{V''})| \geq l$.  Since $\bigcup_{V'' \in \eligible{l}{V'}, v \in V''} B^{\forall}_{V''} 
	= B^{\forall}_{\{v\}}$ by Fact 1 and Proposition 1,   
	it follows that 
	$|W \cap B^{\forall}_{\{v\}}| \geq l$.   
	\textbf{If}: Since $W$ provides $\mathsf{ArgEJR}$-$\mathsf{Spot}$ with 
	respect to $h$, 
	for any $(l, h_1(l), h_2(l))$-cohesive subset $V'$ of $V$, 
	there is some $V'' \in \eligible{l}{V'}$ such that 
	$|W \cap B^{\forall}_{V''}| \geq l$. 
	By Fact 1 and Proposition 1, it follows 
	that $|W \cap B^{\forall}_{\{v\}}| \geq l$ 
	for any $v \in V''$. So, let $v$ be a member of $V''$. 
	By Fact 1 and Proposition 1, it follows that 
	$|W \cap (\bigcup_{V''_o \in \eligible{l}{V'}, v \in V''_o} B^{\forall}_{V''_o})| \geq l$.  

	Now for: $W \in EJR(x)$ iff $W \in ArgEJRS^h(\tau(x))$, 
	\textbf{Only if}: For each $1 \leq l \leq k$, 
	let $\Gamma(l)$ denote the set of all non-empty subsets $V'$ of 
	$V$ such that $V'$ is $l$-cohesive in $x$ and that 
	there is no strict non-empty subset $V''$ of $V'$ 
	such that $V''$ is $l$-cohesive in $x$.  
	Since $W$ provides $\mathsf{EJR}$, 
	for each $1 \leq l \leq k$ 
	and each $V' \in \Gamma(l)$, 
	there is some $v \in V'$ such that $|W \cap A^{\forall}_{\{v\}}| \geq l$.  
	For any $l$-cohesive subset $V'_o$ in $x$ not in $\Gamma(l)$, 
	there is some $V' \in \Gamma(l)$ such that $V' \subset V'_o$. 
	
	By Lemma \ref{lem_relationship_between_cohesion_argumentative_cohesion}, 
	$V'$ is $(l, l, l)$-cohesive in $\tau(x)$ and that 
	there is no strict non-empty subset $V''$ of $V'$ such that 
	$V''$ is argumentative $l$-cohesive in $\tau(x)$. 
	Since $\{v\} \in \eligible{l}{V'}$, 
	$|W \cap B^{\forall}_{\{v\}}| \geq l$.   
	By Proposition \ref{prop_cohesion_monotonicity}, 
	$V'$ is $(l, h_1(l), h_2(l))$-cohesive in $\tau(x)$.  
	By Proposition \ref{prop_monotonicity_argumentative_cohesion}, 
	for any superset $V_{super}$ of $V'$,  $V_{super}$ is argumentative 
	$l$-cohesive. If $V_{super}$ is $(l, h_1(l), h_2(l))$-cohesive, 
	$\{v\} \in \eligible{l}{V_{super}}$; otherwise, there is nothing to show. 

	\textbf{If}: 
	For each $1 \leq l \leq k$, let $\Delta(h, l)$ denote the set of 
	all non-empty subsets $V'$ of $V$ such that $V'$ is $(l, h_1(l), h_2(l))$-cohesive 
	in $\tau(x)$ and that there is no strict non-empty subset $V''$ 
	of $V'$ such that $V''$ is $(l, h_1(l), h_2(l))$-cohesive in $\tau(x)$. 
	Since $W$ provides $\mathsf{ArgEJR}$-$\mathsf{Spot}$ with respect to $h$,  
	for each $1 \leq l \leq k$ 
	and each $V' \in \Delta(h,l)$, there is some $V'' \in 
	\eligible{l}{V'}$ such that $|W \cap B^{\forall}_{V''}| \geq l$. 
	Since, by Fact 1 and Proposition 1, $V'$ is $(l, l, l)$-cohesive in $\tau(x)$, 
	there is in particular some $v \in V''$ such that $|W \cap B^{\forall}_{\{v\}}| 
	\geq l$. For any $(l, h_1(l), h_2(l))$-cohesive subset $V'_o$ 
	in $\tau(x)$ not in $\Delta(h,l)$, 
	there is some $V' \in \Delta(h,l)$ such that $V' \subset V'_o$ 
	and that $v \in \eligible{l}{V'_o}$ for each $v \in V'$. 
	 
	By Lemma \ref{lem_relationship_between_cohesion_argumentative_cohesion}, 
	$V'$ is $l$-cohesive in $x$ and that there is no strict non-empty 
	subset $V''$ of $V'$ such that $V''$ is $l$-cohesive in $x$. 
	Trivially, $|W \cap A^{\forall}_{\{v\}}| \geq l$.   
	For any superset $V_{super}$ of $V'$,  
	$v \in V_{super}$.  
	\hfill$\Box$ \\

\noindent 
\textbf{Proof of Theorem \ref{thm_set_inclusion_representations}}. 
Let $W$ be a member of $ArgEJRS^h(x)$. 
	By assumption, for any $(l, h_1(l), h_2(l))$-cohesive subset $V'$ of $V$, 
	there is some $V'' \in \eligible{l}{V'}$ such that 
	$|W \cap B_{V''}^{\forall}| \geq l$.    
	Implicitly, $|\eligible{l}{V'}| \geq 1$.  
	Let $v$ be a member of $V''$. 
	Since  
	$B_{V''}^{\forall} \subseteq \bigcup_{V_x \in \eligible{l}{V'}, v \in V_x} 
	B_{V_x}^{\forall}$, 
	it holds that $W \in ArgEJR^h(x)$, as required.  
	
	Next, let $W$ be a member of $ArgEJR^h(x)$. It is straightforward to see that 
	$W \in ArgPJR^h(x)$. 

	Next, let $W$ be a member of $ArgPJR^h(x)$. It is straightforward to see that 
	$W \in ArgJR^h(x)$. 
	\hfill$\Box$\\ 

\noindent 
\textbf{Proof of Theorem \ref{thm_separations}}. 
  As we saw in section 4.1, Example 1 is the witness.  
	But also, in general, 
	the second and the third are immediate from Theorem \ref{thm_preservation_m_m_representations} and the known results that $\mathsf{JR}$, $\mathsf{PJR}$ and $\mathsf{EJR}$ 
	do not coincide. This leaves the first obligation still to verify.  
	 
	Let $x$ be $(\{v_1,v_2\}, \{c_1,\ldots,c_8\}, B, 3)$ 
	where $B_{\{v_1\}}$ and $B_{\{v_2\}}$ are as shown below. 
	 
	 \[
\begin{tikzpicture}[
  baseline,
  box/.style={
    draw,
    rounded corners,
    inner sep=8pt
  }
]

\node[box] (B1) {
\begin{tikzcd}[row sep=small,column sep=small]
	c_1 \arrow[d] & c_8\\ 
	c_2 \arrow[loop left] \arrow[r] \arrow[ddr,bend right=10] & c_3   \\
	c_4 \arrow[loop left] \arrow[r] \arrow[dr, bend right=15] & c_5  \\ 
	c_6 \arrow[loop left] \arrow[ur] & c_7 
\end{tikzcd}
};

\node[above left=2pt and 2pt of B1.north west, anchor=south west] 
	{\(\boldsymbol{B_{\{v_1\}}}\)};

\node[box, right=2mm of B1] (B2) {
\begin{tikzcd}[row sep=small,column sep=small] 
	c_1 \arrow[dd,bend right=90] & c_8 \\ 
	c_2 \arrow[loop left] \arrow[r] \arrow[ddr,bend right=10] & c_3   \\
	c_4 \arrow[loop left]\arrow[dr, bend right=15] \arrow[r] & c_5  \\ 
	c_6 \arrow[loop left] \arrow[ruu]  & c_7 
\end{tikzcd}
};

\node[above left=2pt and 2pt of B2.north west, anchor=south west] 
	{\(\boldsymbol{B_{\{v_2\}}}\)};

\end{tikzpicture}
\] 

	We have: $B_{\{v_1\}}^{\forall} = \{c_1, c_3,c_8\}$, 
	$B_{\{v_2\}}^{\forall} = \{c_1, c_5,c_8\}$, 
	$B_{\{v_1,v_2\}}^{\forall} = \{c_1, c_7,c_8\}$.  

	So, $\{v_1\}, \{v_2\}$ and $\{v_1, v_2\}$ are 1-eligible, 2-eligible and 3-eligible. 
	Hare Quota is $|\{v_1, v_2\}|/3$, so $\{v_1\}$ and $\{v_2\}$ 
	are argumentative 1-cohesive, and $\{v_1, v_2\}$ is 
	argumentative 1-, 2- 
	and 3-cohesive.  
	Let $h$ be the permissive $h$, {\it i.e.} $h(l) = (0,0)$ 
	for every positive integer $l$. Then, 
%	by Proposition \ref{prop_conservation}, 
	$\{v_1\}$ and $\{v_2\}$ are $(1, 0, 0)$-cohesive, 
	and $\{v_1,v_2\}$ is $(1, 0,0)$-, $(2,0,0)$- and 
	$(3,0,0)$-cohesive. 

	Let $W$ be $\{c_1,c_3,c_7\}$. Then, 
	\begin{itemize} 
		\item $|B^{\forall}_{\{v_1,v_2\}} \cap W| = 2$. 
		\item $|B^{\forall}_{\{v_1\}} \cap W| = 2$. 
		\item $|B^{\forall}_{\{v_2\}} \cap W| = 1$. 
	\end{itemize}  
	So, for each subset $V''$ of $\{v_1, v_2\}$, $3 > |B^{\forall}_{V''} \cap W|$ and  $W \not\in ArgEJRS^h(x)$. 
	On the other hand, 
	\begin{itemize}  
		\item $|(\bigcup_{V'' \in \eligible{3}{\{v_1,v_2\}}, v_1 \in V''} B^{\forall}_{V''})
	\cap W| = 3 \geq 3$.  
		\item $|(\bigcup_{V'' \in \eligible{2}{\{v_1,v_2\}}, v_1 \in V''} B^{\forall}_{V''})
	\cap W| = 3 \geq 2$.   
		$|(\bigcup_{V'' \in \eligible{1}{\{v_1,v_2\}}, v_1 \in V''} B^{\forall}_{V''})
	\cap W| = 3 \geq 1$.  
		\item $|(\bigcup_{V'' \in \eligible{1}{\{v_1\}}, v_1 \in V''} B^{\forall}_{V''})
	\cap W| = 2 \geq 1$.  
		\item $|(\bigcup_{V'' \in \eligible{1}{\{v_2\}}, v_2 \in V''} B^{\forall}_{V''})
	\cap W| = 1 \geq 1$.  
	\end{itemize}  
	So, $W \in ArgEJR^h(x)$.  
	\hfill$\Box$ \\ 

\noindent 
\textbf{Proof of Theorem \ref{thm_existence_ajr}}. 
{\it \textbf{Proof strategy}. 
	We establish existence directly through an inductive 
	covering argument, rather than by identifying 
	a size-$k$ winner set that optimises a 
	preselected scoring objective. 
	The induction maintains a residual-capacity 
	invariant relating the remaining seats $(\leq k$) to 
	the voters whose cohesive demands are not yet covered. 
	The same proof strategy is used 
	in Theorem \ref{thm_existence_el}.}

We say that a subset $C'$ of $C$ {\it covers} a subset $V'_x$ of $V$ 
	iff, for any $(1, h_1(1), h_2(1))$-cohesive subset 
	$V'$ of $V$, 
	if $V' \cap V'_x \not= \emptyset$, then 
	there is some $V'' \in \eligible{1}{V'}$ such that  
	$V'' \cap V'_x \not= \emptyset$ and that 
	$1 \leq |B^{\forall}_{V''} \cap C'|$. 
	By $\mathsf{cover}(C')$ ($C' \subseteq C$), 
	we denote the largest subset $V'$ of $V$ covered by $C'$. 
 	
	 For any subset $C'$ of $C$, let {\it cohesion rank} of $C'$ be: 
	 \begin{itemize} 
		 \item 1 if some subset of 
			 $(V\backslash \mathsf{cover}(C'))$ 
			 is $(1, h_1(1),h_2(1))$-cohesive; 
		\item 0 if there is no such subset, 
	 \end{itemize} 
	and let {\it residual} of $C'$ be $|V \backslash \mathsf{cover}(C')|$. 

	For any subset $C'$ of $C$, we prove by induction  
	on cohesion rank of $C'$ and subinduction on residual of $C'$ 
	that, if $(k-|C'|)\cdot (|V|/k) \geq \residual(C')$, then the claim holds. %then the claim holds. 

	\textbf{Base case}: Cohesion rank of $C'$ is 0. Then, necessarily, 
	residual of $C'$ is 0. $W$ is any size-$k$ subset of $C$ such that $C' \subseteq W$. 

	\textbf{Inductive case}:  
	Assume that the claim holds for any subset $C'$ of $C$ 
	with $(k - |C'|) \cdot (|V|/k) \geq \residual(C')$ 
	for cohesion rank of $0$. 
	We prove it for any subset $C'$ of $C$ 
	with $(k - |C'|) \cdot (|V|/k) \geq \residual(C')$ 
	and cohesion rank of $1$. By the definition of residual, 
	necessarily $\residual(C') \geq 1 \cdot (|V|/k)$. 
	Then, necessarily $(k-|C'|) \cdot (|V| /k) \geq 
	1 \cdot (|V|/k)$, so necessarily $k \geq |C'| + 1$.  

	Let $V'$ be an $(1, h_1(1), h_2(1))$-cohesive 
	subset of $(V \backslash 
	\mathsf{cover}(C'))$.   
	By the definition of argumentative $1$-cohesion, 
	there exists a subset $V''$ of $V'$ such that 
	$V''$ is $(1,h_1(1),h_2(1))$-cohesive 
	and $1$-eligible. 
	By Proposition \ref{prop_lifting}, 
	there exists some $v \in V''$ such that 
	$|B_{\{v\}}^{\forall} \cap B^{\forall}_{V''}|
	\geq 1$.
	
	Let $c$ be a  member of $B^{\forall}_{\{v\}} 
	\cap B^{\forall}_{V''}$ 
	and let $C''$ be $C' \cup \{c\}$. 
	Then,  $|\mathsf{cover}(C'')| \geq 
	|\mathsf{cover}(C')| + 1 \cdot (|V|/k)$. 
	It holds that $\residual(C'') \leq \residual(C') - 
	1 \cdot (|V|/k)$. 
	In the meantime, $(k-|C''|) \geq (k - |C'|) - 1$. 
	So, $(k - |C''|) \cdot (|V|/k) \geq  
	\residual(C') - 1 \cdot (|V|/k) \geq 
	\residual(C'')$. If cohesion rank of $C''$ is 0, 
	we apply induction hypothesis of main induction. 
	If cohesion rank of $C''$ is $1$, we apply induction 
	hypothesis of  
	subinduction.  
	\hfill$\Box$\\

\noindent 
\textbf{Proof of Theorem \ref{thm_impossibility_zero}}. 
Suppose an argumentative voting profile 
	$(\{v_1, \ldots, v_{12}\}, C, B, 12)$   
	where $B$ is such that it satisfies all the following conditions. 
	\begin{itemize} 
		\item for any $1 \leq i \leq 12$, $B_{\{v_i\}}^{\forall} = \{c_0\}$. 
		\item for any $1 \leq i_1 < i_2 \leq 12$, 
			$B_{\{v_{i_1}, v_{i_2}\}}^{\forall} = \{c_0, c^1_{i_1i_2} 
			\}$.   
		\item for any $1 \leq i_1 < i_2 < i_3 \leq 12$, 
			$B_{\{v_{i_1}, \ldots, v_{i_3}\}}^{\forall}  
			= \{c_0, c^1_{i_1i_2i_3}\}$.   
		\item for any $1 \leq i_1 < i_2 < i_3 < i_4 \leq 12$, 
			$B_{\{v_{i_1}, \ldots, v_{i_4}\}}^{\forall}  
			= \{c_0, c^1_{i_1\ldots i_4}, c^2_{i_1\ldots i_4} 
			\}$.   
		\item for any $1 \leq i_1 < i_2 < i_3 < i_4 < i_5 \leq 12$, 
			$B_{\{v_{i_1}, \ldots, v_{i_5}\}}^{\forall}  
			= \{c_0, c^1_{i_1\ldots i_5}, c^2_{i_1\ldots i_5}\}$.    
		\item for any $1 \leq i_1 < \cdots < i_6 \leq 12$, 
			$B_{\{v_{i_1}, \ldots, v_{i_6}\}}^{\forall}  
			= \{c_0, c^1_{i_1\ldots i_6}, c^2_{i_1\ldots i_6}\}$.    
		\item for any $1 \leq i_1 < \cdots < i_7 \leq 12$, 
			$B_{\{v_{i_1}, \ldots, v_{i_7}\}}^{\forall}  
			= \{c_0, c^1_{i_1\ldots i_7}, c^2_{i_1\ldots i_7}, 
			c^3_{i_1\ldots i_7}\}$.     
		\item for any $1 \leq i_1 < \cdots < i_8 \leq 12$, 
			$B_{\{v_{i_1}, \ldots, v_{i_8}\}}^{\forall}  
			= \{c_0, c^1_{i_1\ldots i_8}, c^2_{i_1\ldots i_8}, 
			c^3_{i_1\ldots i_8}\}$.      
		\item for any $1 \leq i_1 < \cdots < i_9 \leq 12$, 
			$B_{\{v_{i_1}, \ldots, v_{i_9}\}}^{\forall}  
			= \{c_0, c^1_{i_1\ldots i_9}, c^2_{i_1\ldots i_9}, 
			c^3_{i_1\ldots i_9}\}$.     
		\item for any $1 \leq i_1 < \cdots < i_{10} \leq 12$, 
			$B_{\{v_{i_1}, \ldots, v_{i_{10}}\}}^{\forall}  
			= 
			\{c_0, c^1_{i_1\ldots i_{10}}, c^2_{i_1\ldots i_{10}}, 
			c^3_{i_1\ldots i_{10}}, c^4_{i_1\ldots i_{10}}\}$.     
		\item for any $1 \leq i_1 < \cdots < i_{11} \leq 12$, 
			$B_{\{v_{i_1}, \ldots, v_{i_{11}}\}}^{\forall}  
			= 
			\{c_0, c^1_{i_1\ldots i_{11}}, c^2_{i_1\ldots i_{11}}, 
			c^3_{i_1\ldots i_{11}}, c^4_{i_1\ldots i_{11}}\}$.     
		\item $B_{\{v_{i_1}, \ldots, v_{i_{12}}\}}^{\forall}  
			= \{c_0, c^1_{i_1\ldots i_{12}}, c^2_{i_1\ldots i_{12}}, 
			c^3_{i_1\ldots i_{12}}, c^4_{i_1\ldots i_{12}}\}$.     
	\end{itemize}  
	By Lemmas \ref{lem_collective_approval_gadgets} and \ref{lem_non_approval_gadgets}, this $B$ is constructible. 
	$C$ contains all that appear above, and all that are needed 
	to construct $B$ with collective approval gadgets and non-approval gadgets. 
		
	Let $h$ be the permissive $h$, {\it i.e.} 
       $h_1(l) = 0$ and $h_2(l) = 0$. 
	Hare Quota is $|\{v_1, \ldots, v_{12}\}|/12 = 1$, 
	so, if a subset $W$ of 
	$C$ is to provide $\mathsf{ArgPJR}$, 
	then: 
	\begin{enumerate} 
		\item $|W \cap B^{\forall}_{\{v_1\}}|$ must be at least 1. So $W$ 
			must contain $c_0$. 
		\item  ${|W \cap (\bigcup_{V'' \in \eligible{2}{\{v_1,v_2,v_3\}}}B^{\forall}_{V''}})|$, 
			${|W \cap (\bigcup_{V'' \in \eligible{2}{\{v_4,v_5,v_6\}}}B^{\forall}_{V''}})|$, 
			${|W \cap (\bigcup_{V'' \in \eligible{2}{\{v_7,v_8,v_9\}}}B^{\forall}_{V''}})|$ and 
			${|W \cap (\bigcup_{V'' \in \eligible{2}{\{v_{10},v_{11},v_{12}\}}}B^{\forall}_{V''})|}$ 
			must be at least 2. Since 
			the four sets $\{v_1,v_2,v_3\},\{v_4,v_5,v_6\},
			\{v_7,v_8,v_9\},\{v_{10},v_{11},v_{12}\}$  
			do not share size-2 or size-3 subsets, and 
			since no size-1 subsets are 2-eligible,
			$W$ must contain 
			one of
			$\{c^1_{i_1i_2},c^1_{i_1i_3},c^1_{i_2i_3},c^1_{i_1i_2i_3}\}$ 
			for the $(2, h_1(2), h_2(2))$-cohesive subset 
			$\{v_1, v_2,v_3\}$. 
			With no loss of generality,\footnote{In this proof, we are 
			showing impossibility of generating a winner set 
			providing $\mathsf{ArgPJR}$, so we  
			mean no generality is lost for that objective.} 
			we assume $W$ 
			contains $c^1_{i_1i_2}$. 
			Similarly for the others, with no loss of generality, we 
			assume $W$ contains $c^1_{i_4i_5}$,
			$c^1_{i_7i_8}$ and $c^1_{i_{10}i_{11}}$. 
		\item ${|W \cap (\bigcup_{V'' \in \eligible{3}{\{v_1,\ldots,v_6\}}}B^{\forall}_{V''}})|$ and  
			${|W \cap (\bigcup_{V'' \in \eligible{3}{\{v_7,\ldots,v_{12}\}}}B^{\forall}_{V''}})|$  must be at least 3. Now,  
			the two sets $\{v_1,\ldots,v_6\}$ and $\{v_7,\ldots,v_{12}\}$ 
			do not share size-4 or size-5 subsets, and 
			no size-1, size-2 or size-3 subset is 3-eligible. 
			With no loss of generality, 
			we assume $W$ contains 
			$c^1_{i_1i_2i_3i_4}$ and $c^1_{i_1i_3i_4i_5}$ 
			for the $(3,h_1(3),h_2(3))$-cohesive set $\{v_1,\ldots, v_6\}$, 
			and $c^1_{i_7i_8i_9i_{10}}$ and 
			$c^1_{i_7i_9i_{10}i_{11}}$ for $\{v_7,\ldots,v_{12}\}$. 
		\item ${|W \cap (\bigcup_{V'' \in \eligible{4}{\{v_1,\ldots,v_9\}}}B^{\forall}_{V''}})|$ must be at least 4.  
			No size-1, size-2, size-3, size-4, size-5 or size-6 
			subset of $\{v_1,\ldots,v_9\}$ is 4-eligible. 
			With no loss of generality, 
			we assume $W$ contains  
			$c^1_{i_1i_2i_3i_4i_5i_6i_7}$, 
			$c^1_{i_1i_3i_4i_5i_6i_7i_8}$ and 
			$c^1_{i_1i_4i_5i_6i_7i_8i_9}$. 
		\item ${|W \cap (\bigcup_{V'' \in \eligible{5}{\{v_1,\ldots,v_{12}\}}}B^{\forall}_{V''}})|$ must be at least 5.  
			No size-1, size-2, size-3, size-4, size-5 or size-6, 
			size-7, size-8 or size-9
			subset of $\{v_1,\ldots,v_{12}\}$ is 5-eligible. 
			With no loss of generality, 
			we assume $W$ contains  
			$c^1_{i_1i_2i_3i_4i_5i_6i_7i_8i_9i_{10}}$, 
			$c^1_{i_1i_3i_4i_5i_6i_7i_8i_9i_{10}i_{11}}$, 
			$c^1_{i_1i_4i_5i_6i_7i_8i_9i_{10}i_{11}i_{12}}$ 
			and 
			$c^1_{i_1i_2i_4i_5i_6i_7i_8i_9i_{10}i_{11}}$. 
	\end{enumerate} 
	$|W| = 16$ is now strictly greater than 12. 
	So, there is no $W$ providing $\mathsf{ArgPJR}$ for this argumentative voting 
	profile and $h$. 

	By contrapositive of Theorem \ref{thm_set_inclusion_representations}, 
	this impossibility result extends to $\mathsf{ArgEJR}$ and $\mathsf{ArgEJR}\text{-}\mathsf{Spot}$. \hfill$\Box$\\ 

\noindent 
\textbf{Proof of Theorem \ref{thm_existence_el}}. 
    We say that a subset $C'$ of $C$ {\it covers} a subset $V'_x$ of $V$ 
	iff, for any $(l, h_1(l), h_2(l))$-cohesive subset 
	$V'$ of $V$, 
	if $V' \cap V'_x \not= \emptyset$, then 
	there is some $V'' \in \eligible{l}{V'}$ such that  
	$V'' \cap V'_x \not= \emptyset$ and that 
	$l \leq |B^{\forall}_{V''} \cap C'|$. 
	By $\mathsf{cover}(C')$ ($C' \subseteq C$), 
	we denote the largest subset $V'$ of $V$ covered by $C'$.

	For any subset $C'$ of $C$, let {\it cohesion rank} of $C'$ be: 
\begin{itemize} 
		 \item the 
	maximum integer $i$ with which some subset of $(V\backslash \mathsf{cover}(C'))$ 
			 is $(i, h_1(i), h_2(i))$-cohesive; 
		\item 0 if there is no such subset, 
	 \end{itemize} 
	and let {\it residual} of $C'$ be $|V \backslash \mathsf{cover}(C')|$. 

	For any subset $C'$ of $C$, we prove by induction  
	on cohesion rank of $C'$ and subinduction on residual of $C'$ 
	that, if $(k-|C'|)\cdot (|V|/k) \geq \residual(C')$, then the claim holds. %then the claim holds. 

	\textbf{Base case}: Cohesion rank of $C'$ is 0. Then, necessarily, 
	residual of $C'$ is 0. $W$ is any size-$k$ subset of $C$ such that 
	$C' \subseteq W$. 

	\textbf{Inductive case}:  
	Assume that the claim holds for any subset $C'$ of $C$ 
	with $(k - |C'|) \cdot (|V|/k) \geq \residual(C')$ 
	up to cohesion rank of $i$. We prove it for any subset $C'$ of $C$ 
	with $(k - |C'|) \cdot (|V|/k) \geq \residual(C')$ 
	and cohesion rank of $i+1$. By the definition of residual, 
	necessarily $\residual(C') \geq (i+1) \cdot (|V|/k)$. 
	Then, necessarily $(k-|C'|) \cdot (|V| /k) \geq 
	(i+1) \cdot (|V|/k)$, so necessarily $k \geq |C'| + (i+1)$. 

	Let $V'$ be an $((i+1), h_1(i+1), h_2(i+1))$-cohesive subset of $(V \backslash 
	\mathsf{cover}(C'))$, 
	and let $c_1, \ldots, c_{i+1}$ be $i+1$ members of 
	$\bigcap_{V'' \in \eligible{i+1}{V'}}B^{\forall}_{V''}$. 
	Since $h_1(i+1) = h_2(i+1) = i+1$, 
	there is some $v \in V'$ such that these $i+1$ members are in   
	$B^{\forall}_{\{v\}}$. Thus, for any $(l, h_1(l), h_2(l))$-cohesive subset 
	$V'_1$ of $V$, 
	if $V'_1$ contains $v$, 
	it holds that there is some $V''_1 \in \eligible{l}{V'_1}$ (namely, 
	$\{v\}$) such that 
	$l \leq |B^{\forall}_{V''_1} \cap W|$.  
        Hence, $|\mathsf{cover}(C' \cup \{c_1, \ldots, c_{i+1}\})| \geq 
	|\mathsf{cover}(C')| + (i+1) \cdot (|V|/k)$. 
	Let $C''$ be $C' \cup \{c_1, \ldots, c_{i+1}\}$. 
	It holds that $\residual(C'') \leq \residual(C') - 
	(i+1) \cdot (|V|/k)$. 
	In the meantime, $(k-|C''|) \geq (k - |C'|) - (i+1)$. 
	So, $(k - |C''|) \cdot (|V|/k) \geq \residual(C') - (i+1) \cdot (|V|/k) 
	\geq \residual(C'')$. If cohesion rank of $C''$ is strictly 
	smaller than $i+1$, we apply induction hypothesis of main induction. 
	If cohesion rank of $C''$ is $i+1$, we apply induction hypothesis of  
	subinduction.  

	Hence, with respect to $h$, $W$ provides 
	$\mathsf{ArgEJR}$-$\mathsf{Spot}$, 
	and by Theorem \ref{thm_set_inclusion_representations}, 
	also $\mathsf{ArgEJR}$ and $\mathsf{ArgPJR}$.  
	\hfill$\Box$\\ 

\noindent 
\textbf{Proof of Theorem \ref{thm_verification_hardness}}. 
  We prove $\mathsf{coNP}$-hardness by reducing 
	\textsc{Clique} to the
complement of the verification problem. 

	Let $G \equiv (U, E)$ 
	be an undirected graph and let $q$ be an integer 
	greater than or equal to 2. We construct, in polynomial
time, an argumentative voting profile $(V,C,B,k)$ and 
	a size-$k$ winner set $W\subseteq C$ 
	such that $G$ has a clique of size $q$
iff $W$ fails to provide the axiom under consideration.

	{\it \pmb{Construction of $(V,C,B,k)$}.}  
Let $k$ be $\lceil |U|/q \rceil$. 
For every $u \in U$, we create a distinct voter $v_u$. 
	We also create $qk - |U|$ voters different from them, 
	so that $|V| = qk$. 
	
	Let $W$ be $\{c_1, \ldots, c_k\}$ and 
	let $c_0$ be such that $c_0 \not\in W$.    
	For each $t \in \{1, \ldots, k\}$, we create 
	a candidate $c^{\circ}_t$. 
	For each $e \equiv \{u, u'\} \not\in E$, we create five 
	candidates $c^{\bullet}_e, c^{\diamond}_{e,u}, c^{\diamond}_{e,u'}, c^{\star}_{e,u}$ and $c^{\star}_{e,u'}$. 
	Also, let $\{v^{\times}_1, \ldots, v^{\times}_{qk-|U|}\}$ 
	be $(V \backslash \{v_{u_1}, \ldots, v_{u_{|U|}}\})$. 
	For each $v^{\times}_j$ ($1 \leq j \leq qk-|U|$), 
	we create a candidate $c_{v^{\times}_j}$.  
	$C$ is then $W \cup \{c_0\} \cup 
	\{c^{\circ}_1, \ldots, c^{\circ}_k\} 
	\cup \bigcup_{e \equiv \{u,u'\} \not\in E}\{
		c^{\bullet}_e, c^{\diamond}_{e,u}, c^{\diamond}_{e,u'}, c^{\star}_{e,u}, c^{\star}_{e,u'}\} \cup  \{c_{v^{\times}_1}, 
	\ldots, c_{v^{\times}_i}\}$. 
   
	Now, for every $v \in V$, every $t \in \{1, \ldots, k\}$ 
	and every $e \equiv \{u, u'\} \not\in E$, 
	we introduce the following attacks in $v$'s ballot: 
	\begin{enumerate} 
		\item $(c^{\circ}_t, c_t)$. So, unless 
			$c^{\circ}_t$ is defeated, $c_t \in W$ 
			is not accepted. 
		\item $(c^{\diamond}_{e,u}, c^{\bullet}_e)$, 
			$(c^{\diamond}_{e,u'}, c^{\bullet}_e)$ 
			and 
			$(c^{\bullet}_e, c^{\circ}_t)$. So, 
			every ballot contains  
			$|\{\{u,u'\} \mid \{u,u'\} \not\in E\}|$ 
			attackers 
			of $c^{\circ}_t$ and two attackers 
			for each of the attackers. 
	\end{enumerate} 
	Second, for every $e \equiv \{u, u'\} \not\in E$, 
	we introduce $(c^{\star}_{e,u}, c^{\diamond}_{e,u})$ 
	in $v_u$'s ballot, 
	and $(c^{\star}_{e,u'}, c^{\diamond}_{e,u'})$ in $v_{u'}$'s 
	ballot. So, in a combined ballot of voters 
	including $v_u$ and $v_{u'}$, 
	$c^{\bullet}_{e}$ is accepted, and thus all candidates in $W$ 
	are accepted.

	\noindent Finally, for every $v^{\times}_j$ ($1 \leq j 
	\leq qk-|U|$) 
	and every $t \in \{1, \ldots, k\}$, 
	we introduce $(c_{v^{\times}_j}, c^{\circ}_t)$ 
	in $v^{\times}_j$'s ballot. So, 
	in a combined ballot of voters including 
	any $v^{\times}_j$, all candidates in $W$ are accepted. 
	$B$ comprises all these ballots, 
	and this concludes the polynomial-time construction 
	of $(V,C,B,k)$.   

	{\it \pmb{Correctness of the reduction}.}  
	Let us first show that, for any non-empty 
	subset $V'$ of $V$, $W \subseteq B^{\forall}_{V'}$ 
	iff there is some $j \in \{1, \ldots, qk-|U|\}$ such that 
	$v^{\times}_j \in V'$ or there is some 
	$e \equiv \{u,u'\} \not\in E$ such that 
	$\{v_u, v_{u'}\} \subseteq V'$. 
	\textbf{If}: vacuous from the construction.  
	\textbf{Only if}: we show the contrapositive. Suppose 
	there is 
	no $j \in \{1, \ldots, qk-|U|\}$ such that 
	$v^{\times}_j \in V'$ and 
	no $e \equiv \{u,u'\} \not\in E$ such that 
	$\{v_u, v_{u'}\} \subseteq V'$. 
Then, for every \(e=\{u,u'\}\notin E\), it holds that
\(c^\bullet_e\) is not accepted in \(B^\forall_{V'}\), and that at
least one attacker of \(c^\bullet_e\) is in \(B^\forall_{V'}\).
Hence, for every \(t\in\{1,\ldots,k\}\), every attacker
\(c^\bullet_e\) of \(c^\circ_t\) is attacked by some accepted candidate.
Therefore \(c^\circ_t\in B^\forall_{V'}\). Since \(c^\circ_t\) attacks
\(c_t\), it follows that \(c_t\notin B^\forall_{V'}\).

We now show correctness of the reduction. 
\textbf{Only if}:  
Suppose $G$ has a size-$q$ clique $Q\subseteq U$.  
Let $V'_Q$ be $\{v_u \mid u \in Q\}$, 
then for every non-empty subset $V''$ of $V'_Q$, 
$c_0 \in B^{\forall}_{V''}$ and  
$V''$ is 1-eligible. Hence, $|\eligible{1}{V'_Q}| = 2^{|V'_Q|}-1 
= 2^q - 1 = 2^{|V|/k} - 1$. Therefore, 
$V'_Q$ is $(1,1,1)$-cohesive. By Proposition 
\ref{prop_cohesion_monotonicity}, $V'_Q$ is $(1, h_1(1), h_2(1))$-cohesive for 
every $h$. 
Now, because $Q$ is a clique, $V'_Q$ 
contains no ${v^{\times}_{j}}$, $j \in \{1, \ldots, qk-|U|\}$, 
and there is no $\{u,u'\} \not\in E$ such that 
$\{v_u, v_{u'}\} \subseteq V'_Q$. Hence, 
for any subset $V''$ of $V'_Q$, it holds that 
$W \cap B^{\forall}_{V''} = \emptyset$. Therefore,  
with respect to any $h$, 
$W$ fails to provide $\mathsf{ArgJR}$. % with respect to $h$. 

For $l=1$, $\mathsf{ArgPJR}$, $\mathsf{ArgEJR}$ and 
$\mathsf{ArgEJR}$-$\mathsf{Spot}$ also require the existence
of at least one candidate in $W$ approved by 
an appropriate eligible
subset. Since no eligible subset of $V'_Q$ approves any candidate in
$W$, with respect to every $h$, $W$ fails to provide  
these axioms. % with respect to any $h$. 

\textbf{If}: We show the contrapositive. 
Suppose that $G$ has no size-$q$ clique. Then, we must show 
that $W$ provides every axiom. Let $l$ be a positive integer 
not greater than $k$ and 
let
$V'$ be a $(l, h_1(l), h_2(l))$-cohesive subset of 
$V$. Since $V'$ is argumentative $l$ cohesive, 
by Proposition \ref{prop_implicit_size_requirement},  
$|V'| \geq l \cdot |V|/k = l \cdot q$.  
We consider cases. 
\begin{description} 
	\item[Case ${v^{\times}_j} \in V'$ for some 
		$j \in \{1, \ldots, qk-|U|\}$]:  
		$W \subseteq B^{\forall}_{\{{v^{\times}_j}\}}$. 
		Since $|W| = k$ and $l \leq k$, 
		$\{{v^{\times}_j}\}$ is $l$-eligible 
		and $|W \cap B^{\forall}_{\{{v^{\times}_j}\}}| 
		\geq l$. 
	\item[Case, otherwise]: Since $|V'| \geq l \cdot q \geq q$ 
		and $G$ has no size-$q$ clique, 
		there exist $u, u' \in U$ such that 
		$\{v_u, v_{u'}\} \subseteq V'$ and that 
		$\{u, u'\} \not\in E$. Since 
		$W \subseteq B^{\forall}_{\{v_u, v_{u'}\}}$, 
		$\{v_u, v_{u'}\}$ is $l$-eligible. 
		Also, $|W \cap B^{\forall}_{\{v_u, v_{u'}\}}| 
		\geq l$. 
\end{description}  
Hence, $W$ provides $\mathsf{ArgEJR}$-$\mathsf{Spot}$ with 
respect to any $h$. By Theorem \ref{thm_set_inclusion_representations}, $W$ provides every axiom with respect to any $h$. 

We have proved that the complement of the verification problem is 
$\mathsf{NP}$-hard. Hence, 
verification is $\mathsf{coNP}$-hard for 
each axiom with respect to any $h$. 
\hfill$\Box$\\ 

\noindent 
\textbf{Proof of Theorem \ref{thm_polynomial_time_construction}}. 
{\it \pmb{Run-time of Algorithm 1}.} 
We assume that each ballot is represented by an adjacency
matrix over \(C\), so that whether \(c_1\) attacks \(c_2\) in a voter's
ballot can be checked in constant time. Alternative standard
representations, such as edge lists, do not affect polynomial-time
computability, although they may change the precise running time.

Set operations over subsets of \(V\) and \(C\) are implemented by lists
or arrays, and therefore take polynomial time in \(|V|\) and \(|C|\).

To construct the initial $waiting$ set, Algorithm~1 computes,
for each candidate \(c\in C\), the largest voter set
$
V_c=\{v\in V\mid \text{\(v\)'s ballot contains no attack against \(c\)}\}.
$
For a fixed pair \((c,v)\), checking whether 
\(v\) attacks \(c\) requires
scanning all possible attackers \(c'\in C\) and testing whether
\((c',c)\in R_v\). This takes \(O(|C|)\) time. Since there are \(|C|\)
candidates and \(|V|\) voters, constructing all sets \(V_c\) takes
$
O(|V||C|^2)
$
time. The size checks \(|V_c|\ge |V|/k\) take at most
\(O(|V||C|)\) time in total, and are therefore dominated by
\(O(|V||C|^2)\).

Next consider an update step. The algorithm works with pairs
\((c,V')\), where \(c\in C\) and \(V'\subseteq V\). For a fixed voter set
\(V'\), computing which candidates are unattacked under the combined
ballot of \(V'\) can be done by scanning all voters in \(V'\) and all
ordered pairs of candidates. This costs at most
$
O(|V||C|^2)
$
time. All other operations in the same update step, such as updating
$out$, $Alloc$, and \(D\), testing set membership,
and comparing priorities, are 
polynomially bounded by $O(|V||C|^2)$. 

It remains to bound the total cost of recomputing $waiting$.
Since Algorithm~1 terminates once $|out|=k$, and since by
construction there are at most \(k-1\) voter sets 
in $keys(D)$, we have
$
|keys(D)|<k
$
throughout the execution.

For each \(V'\in keys(D)\), the pairs that can be newly added to
$waiting$ are computed by inspecting the combined ballot
$
(C,\bigcup_{v\in V'}R_v)
$
and the candidates that remain relevant after taking \(D[V']\) 
into account. This computation can be performed in \(O(|V||C|^2)\) 
time for each fixed \(V'\). Since \(|keys(D)|<k\), 
one recomputation of $waiting$ takes
$O(k|V||C|^2)$ time.

Finally, the number of recomputations of $waiting$ 
is bounded by $k$. 
Let $k'$ denote $|out|$ right before the first execution 
of line 14. It holds that $k' = |keys(D)|$.  
If $k' = 0$, then after the execution of line 14, 
$waiting = \emptyset$ and line 16 is executed. 
Otherwise, by the definition of line 14, 
for each $V' \in keys(D)$ and any pair $(c, V') \in waiting$, 
$c \not\in D[V']$. Thus, if the main loop is executed $k$ 
times, there is at least one $V' \in keys(D)$ 
such that $|D[V']| \geq k$. This means that $|out| = k$. 

Hence the total cost of all recomputations is
$O(k^2|V||C|^2)$. 
Adding the $O(|V||C|^2)$ initialisation cost does not change the
bound. Therefore Algorithm~1 runs in
$
O(k^2|V||C|^2)$ 
time under the present assumption that 
each ballot is represented by an adjacency
matrix over \(C\). With other representations such as 
edge lists, Algorithm 1 still runs in polynomial time.  

The Java implementation 
is in the supplementary material, as the 
\mbox{greedyGrounded} method in 
\mbox{GreedyGroundedDemo.java}, and runs 
in \(O(k^2 |V||C|^2)\). 

{\it \pmb{A better run-time bound}.} 
Note that the $waiting$ set needs to be updated 
only incrementally. After a
pair \((c,V')\) is selected, the incremental implementation recomputes 
only the
candidate pairs associated with the affected voter set \(V'\), 
rather than recomputing $waiting$ from all voter sets in \(keys(D)\).
For a fixed \(V'\), this update costs \(O(|V||C|^2)\). Since at most
\(k\) recomputations are done before \(|out|=k\), 
the total update
cost is \(O(k|V||C|^2)\). Since the initialisation cost is
\(O(|V||C|^2)\), the overall running time is
\(O(k|V||C|^2)\). 

The Java implementation of this incremental version 
is in the supplementary material, as the \mbox{greedyGroundedIncremental} 
method in 
\mbox{GreedyGroundedDemo.java}, and runs 
in \(O(k|V||C|^2)\). 

{\it  \pmb{Correctness of Algorithm 1}.}  
Theorem \ref{thm_existence_ajr} proved 
that a size-$k$ subset $W$ of $C$ providing $\mathsf{ArgJR}$ 
with respect to any $h$ exists for every member of $\mathfrak{M}^{\mathfrak{A}}$. 

	For each candidate $c \in C$, Algorithm 1 computes 
	the largest voter set $V'_c \subseteq V$
whose members do not attack $c$. Hence, 
	if there is a $(1, h_1(1), h_2(1))$-cohesive set 
	of voters, 
	then for 
	every maximal $(1,h_1(1),h_2(1))$-cohesive set of voters 
	$V'$ that is itself 1-eligible, 
	the initial $waiting$ set contains at least one 
	pair $(c, V'_c)$ such that 
	$V'_c = V'$.  
	We divide cases by the size of the initial $waiting$ set. 
	\begin{description} 
		\item[Case $|waiting| \leq k$]:  
			Algorithm 1 adds every $c \in C$ 
			occurring in the members of $waiting$ to $out$. 
			Thus, any size-$k$ subset of $C$ 
			that contains $out$ provides $\mathsf{ArgJR}$ 
			(with respect to any $h$).  
		\item[Case $|waiting| > k$]:   
			Algorithm 1 selects 
			a highest-priority pair $(c, V'_c)$ in 
			$waiting$ 
			each time. With Theorem \ref{thm_existence_ajr}'s {\it cover}, $\{c\}$ covers at least 
			$V'_c$ which is at least as large as $|V|/k$.   
			If $\{c\}$ covers $V$, then any size-$k$ 
			subset of $C$ that includes $c$ 
			provides $\mathsf{ArgJR}$ with respect to 
			any $h$. 
			If $\{c\}$ does not cover $V$, 
			there are at least $\lceil |V|/k \rceil$ 
			voters who have not been covered. 
			Necessarily, their scores are 0. 
			The highest-priority pair $(c_x, V'_x)$ 
			in $waiting \backslash \{(c, V'_c)\}$ 
			is therefore such that 
			$|V'_x| \geq |V|/k$. 
		
			This selection process continues, and it 
			takes 
			at most $k$ additions to $out$
			before $out$ covers $V$. \hfill$\Box$\\ 
	\end{description}

\subsection{A3 Claims and Proofs of Auxiliary Lemmas}  

\begin{lemma}[Collective approval gadgets] \label{lem_collective_approval_gadgets} 
	Let an argumentative voting profile $(V, C, B, k)$  
	be such that it satisfies the following conditions. 
	\begin{enumerate} 
		\item $V$ has a non-empty subset $V' \equiv 
			\{v_1, \ldots, v_j\}$. 
		\item $C$ has a subset 
			$\{c_0, c_{v_1}, \ldots, c_{v_j},
			c_{V'}\}$. 
		\item $B$ satisfies the following conditions. 
			\begin{enumerate} 
				\item $c_0 \in \bigcap_{v \in V'} 
					B^{\forall}_{\{v\}}$.
				\item for any $x, y
					\in \{1, \ldots, j\}$,
				    $(c_{v_x}, c_{v_x}) \in 
					 R_{v_y}$ 
					 and 
				     $(c_{v_x}, c_{V'}) \in 
					R_{v_y}$. 
				\item for any $y \in 
					\{1, \ldots, j\}$, 
					$(c_0, c_{v_y}) \in 
					R_{v_y}$ 
					and 
					$(c_0, c_{v_y}) 
					\not\in \bigcup_{v 
					\in (V \backslash \{v_y\})}
					R_{v}$. 
				\item for any $y \in 
					\{1, \ldots, j\}$, 
				      no other attacks 
				      among members of 
					$\{c_0, c_{v_1}, 
				        \ldots, c_{v_j}, c_{V'}\}$ 
					exist in $R_{v_y}$. 
			\end{enumerate} 
	\end{enumerate} 
       For any non-empty subset $V'_x$ of $V$, 
	if $c_{V'} \in B^{\forall}_{V'_x}$, then 
	 $V' \subseteq V'_x$. 
\end{lemma} 
\begin{proof} 
	$B_{\{v_1\}}$ is shown below for $\{c_0, 
	c_{v_1}, \ldots, c_{v_j}, c_{V'}\}$. 
	\[
\begin{tikzpicture}[
  baseline,
  box/.style={
    draw,
    rounded corners,
    inner sep=8pt
  }
]

\node[box] (B1) {
\begin{tikzcd}[row sep=small,column sep=small]
	& c_{V'} \\ 
	c_{v_1} \arrow[loop left] \arrow[ur] & \cdots \arrow[u] & c_{v_j} 
	\arrow[loop right] \arrow[ul] \\ 
	& 	c_0 \arrow[ul]
\end{tikzcd}
};

\node[above left=2pt and 2pt of B1.north west, anchor=south west] 
	{\(\boldsymbol{B_{\{v_1\}}}\)};

\end{tikzpicture}
\] 
If $c_{V'}$ is approved by $V'_x$, it has to be 
	that $V' \subseteq V'_x$. 
\end{proof} 
\begin{lemma}[Non-approval gadgets] \label{lem_non_approval_gadgets} 
	Let an argumentative voting profile $(V, C, B, k)$ 
	be such that it satisfies the following conditions.  
	\begin{enumerate} 
		\item  $V$ has a non-empty subset $V'_1$ 
			such that $V'_2 \equiv (V \backslash V'_1)$ 
			is non-empty. 
		\item  $C$ has a subset $\{c_0, c_1\}$. 
		\item  for any $v \in V$,   
			$(c_0, c_0) \in R_v$. 
		\item for any $v \in V'_1$, 
			$(c_0, c_1) \in R_v$.  
		\item for any $v \in V$, 
			no other attacks among 
			$\{c_0, c_1\}$ exist in $R_v$. 
	\end{enumerate}  
	For any non-empty subset $V'$ of $V$, 
	if $V' \cap V'_1 \not= \emptyset$, 
	then $c_1 \not\in B^{\forall}_{V'}$. 
\end{lemma} 
\begin{proof}  
$B_{\{v'\}}$ for $v' \in V'_1$ and 
	$B_{\{v''\}}$ for $v'' \in V'_2$ 
are shown below. 
       \[
\begin{tikzpicture}[
  baseline,
  box/.style={
    draw,
    rounded corners,
    inner sep=8pt
  }
]

\node[box] (B1) {
\begin{tikzcd}[row sep=small,column sep=small]
	 c_1 \\  
	 c_0 \arrow[loop left] \arrow[u]
\end{tikzcd}
};

\node[above left=2pt and 2pt of B1.north west, anchor=south west] 
	{\(\boldsymbol{B_{\{v'\}}}\)};

\node[box, right=20mm of B1] (B2) {
\begin{tikzcd}[row sep=small,column sep=small]
	 c_{1} \\ 
	c_{0} \arrow[loop left] %& \cdots \arrow[u] & c_{v_j} 
%	\arrow[loop right] \arrow[ul] \\ 
%	& 	c_0 \arrow[ul]
\end{tikzcd}
};

\node[above left=2pt and 2pt of B2.north west, anchor=south west] 
	{\(\boldsymbol{B_{\{v''\}}}\)};

\end{tikzpicture}
\]  
	Since $c_0 \not\in B^{\forall}_{V'}$ for any $V' \subseteq V$,
	$v' \in V'$ implies that $c_1 \not\in B^{\forall}_{V'}$.  
\end{proof}

\begin{lemma}[Cohesion correspondence in $\tau$]\label{lem_relationship_between_cohesion_argumentative_cohesion} 
	Let $x \equiv (V,C,A,k)$ be a member of 
	$\mathfrak{M}^{\mathfrak{M}}$. 
	For any non-empty subset $V'$ of $V$ and any positive integer $l$, 
	\begin{enumerate} 
		\item If $V'$ is $l$-cohesive in $x$, then 
			$V'$ is $(l, l, l)$-cohesive in $\tau(x)$. 
		\item If $V'$ is argumentative $l$-cohesive 
			in $\tau(x)$, then there is 
			a non-empty subset $V''$ of $V'$ such that  
			$V''$ is argumentative $l$-cohesive in 
			$\tau(x)$ and 
			$l$-cohesive in $x$. 
		\item If $V'$ is argumentative $l$-cohesive 
			in $\tau(x)$ and $l$-cohesive 
			in $x$, then for any non-empty subset $V''$ of $V'$, 
			if $V''$ is argumentative $l$-cohesive 
			in $\tau(x)$, then
			$V''$ is $l$-cohesive in $x$. 
	\end{enumerate} 
\end{lemma} 
\begin{proof} 
For the first obligation, if $V'$ is $l$-cohesive in $(V,C,A,k)$, 
	then 
	%there is some $v \in V'$ such that $|W \cap A^{\forall}_{\{v\}}| \geq 1$.   
	%Further, 
	$|V'| \geq l \cdot (|V|/k)$ and $|A^{\forall}_{V'}| \geq l$. 
	By Fact 1, for any non-empty subset $V''$ of $V'$, 
	it holds that $A^{\forall}_{V'} \subseteq A^{\forall}_{V''}$. 
	So, $\eligible{l}{V'}$ comprises all non-empty subsets of $V'$.  
	Then, $|\eligible{l}{V'}| \geq 2^{l \cdot (|V|/k)} - 1$ 
	and $l \leq |\bigcap_{V'' \in \eligible{l}{V'}} B^{\forall}_{V''}|$ 
	and $l \leq min_{\emptyset \not= V'' \subseteq V'} |B^{\forall}_{V''}|$. 
	So, $V'$ is $(l, l, l)$-cohesive. 

	For the second obligation,  
	 suppose $V'$ is argumentative $l$-cohesive in $\tau(V,C,A,k)$. 
	 If $V'$ is $(l, l, l)$-cohesive, then clearly $V'$ is $l$-cohesive 
	 in $(V,C,A,k)$. 
	 Otherwise, 
	 there is a non-empty subset $V''_x$ of $V'$ such that 
	 $(V' \backslash V''_x)$ is $(l, l, l)$-cohesive in $\tau(V,C,A,k)$ 
	 but $V' \cup \{v\}$ is not $(l, l, l)$-cohesive in $\tau(V,C,A,k)$ 
	 for each $v \in V''_x$. Let $V''$ denote $(V' \backslash V''_x)$, 
	 then $V''$ is $l$-cohesive in $(V,C,A,k)$.  

	 For the third obligation, 
	 suppose $V'$ is argumentative $l$-cohesive in $\tau(V,C,A,k)$ 
	 and $l$-cohesive in $(V,C,A,k)$.  
	 Since $V'$ is then $(l, l, l)$-cohesive in $\tau(V,C,A,k)$, 
	 for any subset $V''$ of $V'$, 
	 if $V''$ is argumentative $l$-cohesive 
	 in $\tau(V,C,A,k)$, then 
	 $V''$ is necessarily $l$-cohesive in $(V,C,A,k)$. 
\end{proof} 

\end{document}